\documentclass[a4paper,11pt]{article}
\pdfoutput=1
\usepackage[UKenglish]{babel}
\usepackage{a4wide}
\usepackage[noadjust]{cite}
\usepackage{amsmath}
\usepackage{amsfonts}
\usepackage{amsthm}
\usepackage{amssymb}
\usepackage{mathtools}
\usepackage{microtype}
\usepackage{bm}
\usepackage{thm-restate}
\usepackage{hyperref}
\hypersetup{colorlinks=true,citecolor=blue,linkcolor=blue,urlcolor=blue}
\usepackage{cleveref}
\usepackage{enumitem}
\usepackage{complexity}
\usepackage{mathrsfs}

\usepackage{graphicx} 

\let\orignot\not 
\usepackage{mathabx}  
\let\not\orignot

\newcommand{\mypar}[1]{\paragraph*{#1}}

\DeclareMathOperator{\pcsp}{PCSP}
\DeclareMathOperator{\csp}{CSP}
\DeclareMathOperator{\pol}{Pol}
\DeclareMathOperator{\Ab}{A}
\DeclareMathOperator{\im}{Im}
\DeclareMathOperator{\id}{id}
\DeclareMathOperator{\Lin}{Eqn}
\DeclareMathOperator{\PLin}{PEqn}

\DeclareMathOperator{\Ls}{L}
\DeclareMathOperator{\Rs}{R}
\DeclareMathOperator{\Cs}{C}
\DeclareMathOperator{\LCs}{LC}
\DeclareMathOperator{\LRs}{LR}
\DeclareMathOperator{\CRs}{CR}

\newcommand{\dom}{\mathrm{dom}}

\newcommand{\substr}{\preceq}

\newcommand{\leqsg}{\sqsubseteq}
\newcommand{\geqsg}{\sqsupseteq}
\newcommand{\lneqsg}{\sqsubsetneq}

\renewcommand{\A}{{\bm{A}}}
\newcommand{\B}{{\bm{B}}}
\renewcommand{\D}{{\bm{D}}}
\renewcommand{\C}{{\bm{C}}}
\newcommand{\I}{{\bm{I}}}
\newcommand{\X}{\bm{X}}

\newcommand{\Mcal}{\mathcal{M}}
\newcommand{\Mscr}{\mathscr{M}}
\newcommand{\NN}{\mathbb{N}}
\newcommand{\arty}{\mathrm{ar}}

\newcommand\Crestrict[2]{{
  \left.\kern-\nulldelimiterspace 
  #1 
  \right|_{#2} 
  }}

\newcommand{\AIP}{\textsc{AIP}}
\newcommand{\BLP}{\textsc{BLP}}

\newcommand{\tuple}[1]{\ensuremath{\mathbf{#1}}}

\newtheorem{lemma}{Lemma}
\newtheorem{corollary}{Corollary}
\newtheorem{observation}{Observation}
\newtheorem{theorem}{Theorem}

\theoremstyle{definition}
\newtheorem{definition}{Definition}
\newtheorem{remark}{Remark}
\newtheorem{example}{Example}

\begin{document}

\author{Alberto Larrauri\\
University of Oxford\\
\texttt{alberto.larrauri@cs.ox.ac.uk}
\and
Stanislav \v{Z}ivn\'y\\
University of Oxford\\
\texttt{standa.zivny@cs.ox.ac.uk}
}

\title{Solving promise equations over monoids and groups\thanks{An extended abstract of this work appeared in the Proceedings of the ICALP 2024~\cite{Larrauri24:icalp}. This work was supported by UKRI EP/X024431/1. For the purpose of Open Access, the authors have applied a CC BY public copyright licence to any Author Accepted Manuscript version arising from this submission. All data is provided in full in the results section of this paper.}}

\date{\today}
\maketitle

\begin{abstract}
We give a complete complexity classification for the problem of finding a solution to a given system of equations over a fixed finite monoid, given that a solution over a more restricted monoid exists.
As a corollary, we obtain a complexity classification for the same problem over groups.
\end{abstract}

\section{Introduction}
Constraint satisfaction problems (CSPs) form a large class of fundamental computational problems studied in artificial intelligence, database theory, logic, graph theory, and computational complexity. 
Since CSPs (with infinite domains) capture,  
up to polynomial-time Turing reductions, \emph{all} computational problems~\cite{Bodirsky08:icalp}, some restrictions need to be imposed on CSPs in order to have a chance to obtain complexity classifications. One line of work, pioneered in the database theory~\cite{Kolaitis00:jcss}, restricts the interactions of the constraints in the instance~\cite{Grohe07:jacm,Marx13:jacm}. 

Another line of work, pioneered in~\cite{Jeavons97:closure,Feder98:monotone}, restricts the types of relations used in the instance; these CSPs are known as nonuniform CSPs, or as having a fixed template/constraint language. Such CSPs with infinite domains capture graph acyclicity, systems of linear equations over the rationals, and many other problems~\cite{bodirsky2021complexity}. Already fixed-template CSPs with finite domains form a large class of fundamental problems, including graph colourings~\cite{HellN90}, variants of the Boolean satisfiability problem, and, more generally, systems of equations over different types of finite algebraic structures. Even then, the class of finite-domain CSPs
avoided a complete complexity classification for two decades despite a sustained effort. 

In 2017, Bulatov~\cite{Bulatov17:focs} and, independently, Zhuk~\cite{Zhuk20:jacm} classified all finite-domain CSPs as either solvable in polynomial time or \NP-hard, thus answering in the affirmative the Feder-Vardi dichotomy conjecture~\cite{Feder98:monotone}.
In the effort to answer the Feder-Vardi conjecture, complexity dichotomies 
were established for restricted 
fragments of CSPs, e.g., conservative CSPs~\cite{Bulatov11:tocl}, and
equations over finite algebraic structures such as groups~\cite{Goldmann02:ic} and monoids~\cite{KTT07:tcs}. In particular, while 
systems of equations\footnote{Some papers use the term a \emph{linear} equation.} over Abelian groups are solvable in polynomial time, they are \NP-hard over non-Abelian groups~\cite{Goldmann02:ic}.

One of the recent research directions in constraint satisfaction that has attracted a lot of attention is the area of \emph{promise} CSPs (PCSPs)~\cite{AGH17,BG21:sicomp,BBKO21}. The idea is that each constraint has two versions, a strong version and a weak version. Given an instance, one is promised that a solution satisfying all strict constraints exists and the goal is to find a solution satisfying all weak constraints, which may be an easier task.
The prototypical example is the approximate graph colouring problem~\cite{GJ76}: Given a $3$-colourable graph, can one find a $6$-colouring? The complexity of this problem is open (but believed to be \NP-hard). 
Despite a flurry of papers on PCSPs, e.g.,~\cite{Ficak19:icalp,AB21,Barto21:stacs,bz22:ic,NZ22,NZ24:talg,Barto22:soda,Atserias22:soda,cz23stoc:ba,BGS23,bgs_robust23stoc,Dalmau24:lics,cz23sicomp:clap}, the PCSP complexity landscape is widely open and unexplored. 
It is not even clear whether a dichotomy should be expected. Even the case of Boolean PCSPs remain open, the state-of-the-art being a dichotomy for Boolean \emph{symmetric} PCSPs~\cite{Ficak19:icalp}. This should be compared with Boolean (non-promise)  CSPs, which were classified by Schaefer in 1978~\cite{Schaefer78:stoc}. Schaefer's tractable cases include the classic and well-known examples of CSPs: equations and graph colouring. Both have been  studied on non-Boolean domains and their complexity is well understood.
However, the complexity of the promise variant of these fundamental problems is open. 
The first problem, graph colouring, leads to the already mentioned approximate graph colouring problem, which is a notorious open problem, despite recent progress~\cite{BBKO21,KOWZ23}. In this paper, we look at the second problem, and study PCSPs capturing systems of equations. 

\mypar{Contributions}
The precise statements of all our main results are presented in~\Cref{sec:results}. 

As our most important contribution, in~\Cref{sec:equations} we establish a complexity dichotomy for PCSPs capturing \emph{promise} systems of equations over finite monoids, and over finite groups as a special case. Perhaps unsurprisingly, the tractability boundary is linked to the notion of Abelianness, just like in the non-promise setting, but the result is non-trivial and requires some care. 
Our main tool for studying the computational complexity of PCSPs is the so-called ``algebraic approach'', relying crucially on the notion of a \emph{polymorphism}. Polymorphisms can be seen as high-dimensional symmetries of a PCSP template and capture the complexity of the underlying computational problem~\cite{BG21:sicomp,BBKO21}. Polymorphisms of a PCSP template form a \emph{minion}~\cite{BBKO21}; that is, a family of functions that is closed under permuting arguments, identifying arguments, and adding dummy arguments. As we shall see later, it is useful to study more abstract minions, not only families of functions, cf.~\cite{BBKO21,BGWZ20}. Following the approach from~\cite{BBKO21}, hardness of a PCSP is established by showing that the associated polymorphism minion is, in some sense, limited. Conversely, if this minion is rich enough then the PCSP can be shown to be solvable via some efficient algorithm~\cite{BBKO21,BGWZ20,cz23sicomp:clap,bgs_robust23stoc}.

To prove our main result, we study a class of minions that arise naturally from monoids, which we call monoidal minions. In~\Cref{sec:monoidal} 
we show a complexity dichotomy for PCSPs whose polymorphism minions are homomorphically equivalent to some monoidal minion. This is our second contribution, which may be of independent interest.
In particular, the concept of monoidal minions captures studied minions, cf.~\Cref{rem:minions} in~\Cref{sec:results}.

All our tractability results use solvability by the $\BLP+\AIP$ algorithm~\cite{BGWZ20}. In fact, tractable PCSPs corresponding to promise systems of equations over monoids are finitely tractable in the sense of~\cite{BG21:sicomp,AB21}. In the special case of promise systems of equations over groups, the affine integer programming ($\AIP$) algorithm~\cite{BG21:sicomp,BBKO21} suffices, rather than $\BLP+\AIP$. However, $\AIP$ is provably not enough to solve promise equations over general monoids.  

As our final contribution, in~\Cref{sec:semigroups} we show that our dichotomy for systems of equations over monoids cannot be easily extended to semigroups, as this would imply a dichotomy for all PCSPs. We do so by showing that every PCSP is polynomial-time equivalent to a
PCSP capturing systems of equations over semigroups, a phenomenon observed  
for CSPs in~\cite{KTT07:tcs}.

\mypar{Related work}
PCSPs are a qualitative approximation of CSPs; the goal is still to satisfy all constraints, but in a weaker form. A recent related line of work includes the series~\cite{Bhangale21:stoc,Bhangale23:stocII,Bhangale23:stocIII}.
A traditional approach to approximation is quantitative: maximising the number of satisfied constraints. Regarding approximation of equations, 
H{\aa}stad showed that, for any Abelian group $G$ and any $\varepsilon>0$,
it is \NP-hard to find a solution satisfying $1/|G|+\varepsilon$
constraints~\cite{Hastad01} even if $1-\varepsilon$ constraints can be satisfied. Hence, the random assignment, which satisfies $1/|G|$ constraints, is optimal!
H{\aa}stad's result has been extended to non-Abelian groups
in~\cite{Engebretsen04:tcs,Bhangale21:stoc}.
Systems of equations have been studied, e.g., over semigroups in~\cite{Seif03:jc},
over monoids and semigroups in~\cite{KTT07:tcs}, and over arbitrary finite
algebras in~\cite{Larose06:ijac,kompatscher2018equation,Bodirsky21:jml,Mayr23:mfcs}.

The work of Nakajima and \v{Z}ivn\'y on symmetric functional
PCSPs~\cite{NZ24:talg} is somewhat related to PCSPs and equations but is incomparable to the work in the present article. 
Since we do not need it, we do not use the language of category theory but we remark that minions and other concepts can be presented in a category-theoretical way, cf.~\cite{NZ24:talg,Dalmau24:lics}.

\section{Preliminaries}
\label{sec:prelims}

We denote by $[k]$ the set $\{1,2,\ldots,k\}$. We write $\id_X$ for the identity map on a set $X$. We use the lowercase boldface font for tuples; e.g., we write $\bm{b}$ for a tuple $(b_1,\dots, b_n)$. We say that a function $f$ \emph{extends} another function $g$ if $\dom(g)\subseteq \dom(f)$, and $\Crestrict{f}{\dom(g)}=g$.

\mypar{Algebraic structures}
A \emph{semigroup} $S$ is a set equipped with an associative binary operation, for which we use multiplicative notation.
Two elements $a,b\in S$ commute if $ab=ba$. An element $a$ is \emph{idempotent}
if $aa=a$.
An \emph{Abelian} semigroup is a semigroup in which every two elements 
commute.
A \emph{semigroup homomorphism} from a semigroup $S_1$ to a semigroup $S_2$ is a map $\varphi:S_1\to S_2$ satisfying $\varphi(s\cdot_{S_1}t)=\varphi(s)\cdot_{S_2}\varphi(t)$.\footnote{I.e., the multiplication on the LHS is in $S_1$, whereas the multiplication on the RHS is in $S_2$.} 
Given two elements $s,t\in S$ we write $s\leqsg t$ if $s=t$ or there is an element $r\in S$ satisfying $tr=s$. Note that $\leqsg$ constitutes a preorder over any semigroup, i.e., $\leqsg$ is reflexive and transitive. We define the equivalence relation $\sim$ by  $s\sim t$ whenever $s\leqsg t$ and $t\leqsg s$.

A \emph{monoid} is a semigroup  containing an identity element for its binary operation, denoted by $e$.
A \emph{monoid homomorphism} from a monoid $M_1$ to a monoid $M_2$ is a map $\varphi:M_1\to M_2$ satisfying $\varphi(x\cdot_{M_1}y)=\varphi(x)\cdot_{M_2}\varphi(y)$ and $\varphi(e_{M_1})=e_{M_2}$.
We say that $\varphi$ is \emph{Abelian} if its image $\im(\varphi)$ is an Abelian monoid.

A \emph{group} is a monoid in which each element has an inverse. 
A \emph{group homomorphism} 
from a group $G_1$ to a group $G_2$ is a map $\varphi:G_1\to G_2$ satisfying
$\varphi(x\cdot_{G_1}y)=\varphi(x)\cdot_{G_2}\varphi(y)$ (which implies that also the inverses and the identity element are preserved). \par

Given a semigroup $S$, a subset $G\subseteq S$ is called a \emph{subgroup} if $G$ equipped with $S$'s binary operation is a group, meaning that there is a distinguished element $e_G\in G$ satisfying that (1) $e_G\cdot_M g = g\cdot_M e_G = g$ for each $g\in G$, and (2) 
for each element $g\in G$ there exists $h\in G$ satisfying $g\cdot_M h= h \cdot_M g = e_G$. We say that $S$ is a \emph{union of subgroups} if every element $s\in S$ belongs to a subgroup of $S$.

We call an element $s$ of a semigroup $S$ \emph{regular}\footnote{In the
extended abstract of this work~\cite{Larrauri24:icalp}, we required that
$s^2t=s$ and $st=ts$ for some $t\in S$. For a finite semigroup $S$, this is
equivalent to requiring that $s^2t=s$ for some $t\in S$ as if this second
condition holds, then $s^k=s$ for some $k>1$
by~\Cref{lem:regularity}\,(2), implying the existence of $t$ that commutes with $s$ and satisfies $s^2t=s$.} if $s^2t=s$ for some $t$ in $S$.\footnote{The usual definition of a regular element in a semigroup, which is weaker, requires that $sts=s$ for some $t$~\cite{howie1995fundamentals}. What we call regular is often called completely regular.}
 Intuitively, $t$ acts as some type of inverse of $s$. It is known that $s$ belongs to a subgroup of $S$ if and only if $s$ is regular~\cite[Theorem 2.2.5]{howie1995fundamentals}. We will make use of the following equivalent characterisations of regularity.
 
\begin{lemma}\label{lem:regularity}
    Let $S$ be a finite semigroup and $s\in S$.
    Then the following are equivalent:
    \begin{enumerate}
      \item[(1)] \label{it:reg1} $s$ is regular,
      \item[(2)] \label{it:reg2} $s^k=s$ for some $k>1$,
      \item[(3)] \label{it:reg3} $s$ belongs to a subgroup of $S$,
      \item[(4)] \label{it:reg4} $s \leqsg s^2$.
    \end{enumerate}
\end{lemma}

\begin{proof}
\begin{description}
\item[(3) $\implies$ (2):] For any finite group $G$ there exists number $k>1$ such that $g^k=g$ for all $g\in G$. \par
\item[(2) $\implies$ (3):] Consider the set $G = \{ s^\ell \mid 1\leq \ell < k
    \}$. We claim that $G$ is a group whose identity is $s^{k-1}$. By~(2), $s^{k-1}$ acts as a multiplicative identity in $G$. Moreover, given any $1\leq \ell < k-1$, the inverse of $s^\ell$ in $G$ is simply $s^{k-1-\ell}$. \par
    \item[(1) $\implies$ (4):] By the definition of regularity, there is some element $t$ such that $s^2t=s$, meaning that $s\leqsg s^2$. \par
    \item[(2) $\implies$ (1):] If $k=2$, let $t=s$. Otherwise, if $k>2$, let $t=s^{k-2}$. Then we have $s^2t= s$. \par
    \item[(4) $\implies$ (2):] By assumption, $s^2 t = s$ for some $t\in S$. 
    Note that this implies that $s^{k+1} t^k = s$ for all $k\geq 1$. 
    As $S$ is finite, there must be numbers $k > \ell > 1$ satisfying that $s^k= s^\ell$. Then, the following chain of identities holds
    \[
    s = s^k t^{k-1} =  s^\ell t^{k-1} =  s^\ell t^{\ell-1} t^{k-\ell}
     = s t^{k-\ell}.    
    \]
    In particular, this means that $s^{k^\prime} = s^{k^\prime} t^{k-\ell}$ for any $k^\prime\geq 1$.
    Finally, it also holds that $s = s^{k-\ell + 1} t^{k-\ell}$, which together
    with the last equality yields that
    $s= s^{k-\ell + 1}$. Since $\ell<k$, we have
    $k-\ell + 1 > 1$, thus establishing~(2).
 \end{description}   
\end{proof}

\par

We use the standard product (and also the power) of a semigroup (monoid, group), where the operation is defined component-wise.
We use the symbol $\substr$ for a substructure; e.g., if $S$ is a semigroup then we write $T\substr S$ to indicate that $T$ is a subsemigroup of $S$ (and similarly for monoids and groups).

Unless stated explicitly otherwise, all semigroups, monoids, and groups in this paper are finite.
 
\mypar{Relational structures}
A \emph{relational signature} $\sigma$
consists of a finite set of relation symbols $R$, each with a finite arity $\arty(R)\in\NN$.
A \emph{relational structure} $\A$ over the signature $\sigma$,
or a $\sigma$-structure, 
consists of a finite set $A$ and a relation $R^\A\subseteq A^k$ of arity $k=\arty(R)$ for every $R\in\sigma$.
Let $\A$ and $\B$ be two $\sigma$-structures. A map $h:A\to B$ is called a \emph{homomorphism} from $\A$ to $\B$ if $h$ preserves all relations in $\A$; i.e., if, for every $R\in\sigma$, 
$h(\tuple{x})\in R^\B$ whenever $\tuple{x}\in R^\A$, where $h$ is applied component-wise. We denote the existence of a homomorphism from $\A$ to $\B$ by writing $\A\to\B$.
A \emph{template} is a pair $(\A,\B)$ of relational structures such that $\A\to\B$.

A $k$-ary \emph{polymorphism} of a template $(\A,\B)$ over signature $\sigma$ is a map $p:A^k\to B$ that preserves all relations $R^\A$ from $\A$ in the following sense: For any $\arty(R)\times k$ matrix whose columns belong to $R^\A$, applying $p$ row-wise results in a tuple that belongs to $R^\B$.
We denote by $\pol(\A,\B)$ the set of all polymorphisms of $(\A,\B)$.\footnote{Equivalently, $p$ is a polymorphism of $(\A,\B)$ if $p$ is a homomorphism from the $k$-th power of $\A$ to $\B$.}

The $i$-th coordinate of a map $p:A^k\to A$ is called \emph{essential} if there exist
$a_1,\ldots,a_k\in A$ and $a'_i\in A$ such that $p(a_1,\ldots,a_k)\neq
p(a_1,\ldots,a_{i-1},a'_i,a_{i+1},\ldots,a_k)$. A coordinate that is not
essential is called \emph{inessential}.
A map $p:A^k\to A$ is called \emph{idempotent} if $p(x,\ldots,x)=x$.

\mypar{Minions}
A \emph{minion} $\Mscr$ is a collection of sets $\Mscr(n)$, one for each positive number $n$, such that, for each map $\pi: n \rightarrow m$, there is a map $\pi^\Mscr: \Mscr(n) \rightarrow \Mscr(m)$ satisfying (1) $\id_{[n]}^\Mscr= \id_{\Mscr(n)}$
for every $n\geq 10$, and (2)
$\pi^\Mscr \circ \tau^\Mscr = 
(\pi\circ \tau)^\Mscr$ for every pair of suitable maps $\pi, \tau$. When the minion is clear from the context, we write $p^{(\pi)}$ for $\pi^\Mscr(p)$. Elements $p\in \Mscr(n)$ are called $n$-ary or as having \emph{arity} $n$. Whenever $p^{(\pi)}=q$ we say that $q$ is a \emph{minor} of $p$. A \emph{minion homomorphism}
$\xi: \Mscr \to \mathscr{N}$ is a collection of maps $\xi_n:\Mscr(n) \rightarrow
\mathscr{N}(n)$ for each $n\geq 1$ that preserve minor operations; that is, $\xi_m(p^{(\pi)})=(\xi_n,(p))^{(\pi)}$ for every minor $p^{(\pi)}$, where $\pi:[n]\rightarrow [m]$. 
\par

Given a template $(\A,\B)$, its set of polymorphisms $\pol(\A,\B)$ can be
equipped with a minion structure in a natural way: The $n$-ary elements of $\pol(\A, \B)$ are just 
$n$-ary polymorphisms $p: A^n \rightarrow B$. Additionally, given an $n$-ary
polymorphism $p$ and a map $\pi:[n] \rightarrow [m]$, the minor $p^{(\pi)}$ is the polymorphism
$q: A^m \rightarrow B$ given by $(a_1, \dots, a_m) \mapsto p(b_1, \dots, b_n)$,
where $b_i= a_{\pi(i)}$ for each $i\in [n]$.\par

Given a minion $\Mscr$, we define two special types of elements. An element 
$p\in \Mscr(2k+1)$ is called \emph{alternating} if
$p^{(\pi)}=p$ for any permutation $\pi:[2m+1]\to[2m+1]$ that preserves
parity, and $p^{(\pi_1)}=p^{(\pi_2)}$, where for each $i=1,2$ the map $\pi_i$
is given by $1\mapsto i$, $2\mapsto i$ and $j\mapsto j$ for all $j>2$. 
An element $p\in \Mscr(2k+1)$ is called \emph{$2$-block-symmetric}
if the set $[2k+1]$ can be partitioned into two blocks of size $k+1$ and $k$ in such a
way that $p^{(\pi)}=p$ for any map $\pi:[2m+1]\to[2m+1]$ that preserves each block

\mypar{Constraint satisfaction}
Let $(\A,\B)$ be a template with common signature $\sigma$. The \emph{promise constraint satisfaction problem} (PCSP) with template $(\A,\B)$ is the following computational problem, denoted by $\pcsp(\A,\B)$. Given a $\sigma$-structure $\X$, output \textsc{Yes} if $\X\to\A$ and output \textsc{No} if $\X\not\to\B$. This is the decision version. In the search version, one is given a $\sigma$-structure $\X$ with the promise that $\X\to\A$; the goal is to find a homomorphism from $\X$ to $\B$ (which necessarily exists, as $\X\to\A$ and $\A\to\B$, and homomorphisms compose). It is known that the decision version polynomial-time reduces to the search version (but it is not known whether the two variants are polynomial-time equivalent)~\cite{BBKO21}. In our results, the positive (tractability) results are for the search version, whereas the hardness (intractability) results are for the decision version. 
We denote by $\csp(\A)$ the problem $\pcsp(\A,\A)$; this is the standard (non-promise) constraint satisfaction problem (CSP). For CSPs, the decision version and the search version are polynomial-time equivalent~\cite{Bulatov05:classifying}.

We need two existing algorithms for PCSPs, namely the $\AIP$ algorithm~\cite{BBKO21} and the strictly more powerful $\BLP+\AIP$ algorithm~\cite{BGWZ20}. Their power is captured by the following results.
\begin{theorem}[\cite{BBKO21}]\label{thm:aip}
  Let $(\A,\B)$ be a template. Then 
  $\pcsp(\A,\B)$ is solved by $\AIP$ if and only if $\pol(\A,\B)$ contains
  alternating maps of all odd arities.
\end{theorem}

\begin{theorem}[\cite{BGWZ20}]\label{thm:blpaip}
    Let $(\A,\B)$ be a template. Then
    $\pcsp(\A,\B)$ is solved by $\BLP+\AIP$ if and only if $\pol(\A,\B)$ contains $2$-block-symmetric maps of all odd arities.
\end{theorem}

\section{Overview of Results}
\label{sec:results}

\mypar{Promise equations over monoids and groups}
Our first and main result is a dichotomy theorem for solving promise equations over finite monoids and thus also, as a special case, over finite groups. We first define equations in the standard, non-promise setting as it is useful for mentioning previous work and for our own proofs.

An \emph{equation} over a semigroup $S$
is an expression of the form $x_1\dots x_n = y_1 \dots y_m$, where each $x_i, y_i$ is either a variable or some element from $S$, referred to as a \emph{constant}.  A \emph{system of equations} over $S$ is just a set of equations. A solution to such a system is an assignment of elements of $S$ to the variables of the system
that makes all equations hold. 
Equations and systems of equations are defined similarly for monoids and groups. The only difference is that for groups we allow ``inverted variables'' $x^{-1}$ in the equations, which are interpreted as inverses of the elements assigned to $x$.

In the context of CSPs, it is common to consider only restricted ``types'' of equations that can then express all other equations. The following definition captures systems of equations where each equation is either of the form $x_1 x_2= x_3$,
for three variables, or $x=c$, fixing a variable to a constant. It is well known that restricting to systems of equations of this kind is without loss of generality, cf.~\Cref{app:red3}.

\begin{definition}
Let $S$ be a semigroup and $T\substr S$ a subsemigroup.
The relational structure $\Lin(S,T)$  has universe $S$, and the following relations: 
\begin{itemize}
    \item A ternary relation $R_{\times} = \{
    (s_1,s_2,s_3) \in S^3 \mid s_1s_2=s_3 \}$, and
    \item a singleton unary relation $R_t=\{t\}$ for each $t\in T$.
\end{itemize}
\end{definition}

This template captures systems of equations of the kind described above when we
allow only constants in a subsemigroup $T$ of the ambient semigroup $S$.
Similarly, we define the templates $\Lin(M,N)$, $\Lin(G,H)$ in the same way when
$M$ is a monoid and $N\substr M$ a submonoid, and when $G$ is a group and
$H\substr G$ is a subgroup. Observe that the definition of subgroup is more
restrictive than the one of submonoid and this in turn is more restrictive than
the notion of subsemigroup. We slightly abuse the notation and write $\Lin(S,T)$ for $\csp(\Lin(S,T))$. \par

Previous works focused on problems $\Lin(G)=\Lin(G,G)$ and $\Lin(M)=\Lin(M,M)$. Given a group $G$, it is known that $\Lin(G)$ is solvable in polynomial time (by $\AIP$) if $G$ is Abelian, and \NP-hard otherwise~\cite{Goldmann02:ic}. Similarly, when $M$ is a monoid, $\Lin(M)$ is solvable in polynomial time if $M$ is Abelian and it is the union of its subgroups,
and \NP-hard otherwise~\cite{KTT07:tcs}. These results were shown 
before the Dichotomy Theorem for CSPs was proved~\cite{Bulatov17:focs,
Zhuk20:jacm}. The original proofs relied on ad-hoc reductions and various
notions from the theory of groups and the theory of monoids. For the sake of
completeness, we present simplified proofs of those previous results as
corollaries of the Dichotomy Theorem in~\Cref{app:csps}.

We now define promise equations. 

\begin{definition}
\label{def:promise_eqns}
Let $S_1, S_2$ be semigroups, and let $\varphi$ be a semigroup homomorphism 
with $\dom(\varphi)\substr S_1$
and $\im(\varphi)\substr S_2$. 
The \emph{promise system of equations over semigroups} problem
$\PLin(S_1,S_2,\varphi)$ is the $\pcsp(\A,\B)$, where $A=S_1$, $B=S_2$, and the relations are defined as follows:
\begin{itemize}
    \item A ternary relation
    $R_\times^\A=
    \{ (s_1,s_2,s_3)\in S_1^3 \mid s_1s_2=s_3\}$,  and
    $R_\times^\B=
    \{ (s_1,s_2,s_3)\in S_2^3 \mid s_1s_2=s_3\}$.
    \item For each $t\in \dom(\varphi)$, a unary relation given by 
    $R^\A_t=\{t\}$, and
    $R^\B_t=\{ \varphi(t) \}$.
\end{itemize}
For this template to be well defined there should be a homomorphism from $\A$ to $\B$, which is equivalent to the existence of a semigroup homomorphism $\psi: S_1 \rightarrow S_2$
that extends $\varphi$.
\end{definition}

Analogously, we also define
the \emph{promise system of equations over monoids} problem and the
\emph{promise system of equations over groups} problem by replacing semigroup-related notions with monoid-related notions and group-related notions respectively.  
Observe that the problem $\Lin(S,T)$ described before corresponds precisely to $\PLin(S,S, \id_S)$. 

We can now state our main result.

\begin{restatable}[Main]{theorem}
{promisemonoids}\label{th:promisemonoids}
       Let $M_1$, $M_2$ be monoids and $\varphi$ a monoid homomorphism with $\dom(\varphi)\substr M_1, \im(\varphi)\substr M_2$. Then 
       $\PLin(M_1,M_2,\varphi)$ is solvable in polynomial time by $\BLP+\AIP$ if and only if there is an Abelian homomorphism $\psi:M_1\rightarrow M_2$ extending $\varphi$ and $\im(\psi)$ is a union of subgroups.
       If no such homomorphism $\psi$ exists, then $\PLin(M_1,M_2,\varphi)$ is \NP-hard. 
\end{restatable}

\noindent
For the special case of groups, we get a simpler tractability criterion and a simpler algorithm.

\begin{restatable}{corollary}{promisegropus}\label{th:promisegroups}
    Let $G_1$, $G_2$ be groups and $\varphi$ a group homomorphism with $\dom(\varphi)\substr G_1, \im(\varphi)\substr G_2$. Then 
    $\PLin(G_1,G_2, \varphi)$ is solvable in polynomial time by $\AIP$ if and only if there is an Abelian homomorphism $\psi:G_1 \rightarrow G_2$ 
    extending $\varphi$. If no such homomorphism $\psi$ exists, then $\PLin(G_1,G_2,\varphi)$ is \NP-hard. 
\end{restatable}

As easy corollaries, \Cref{th:promisemonoids} applies in the special case of non-promise setting.

\begin{corollary}
Given two monoids $N\substr M$,  $\Lin(M,N)$ is solvable in polynomial time by $\BLP+\AIP$ if and only if there is an Abelian endomorphism of $M$ extending $\id_N$ whose image is a union of subgroups. If no such endomorphism exists, then $\Lin(M,N)$ is \NP-hard.
\end{corollary}

\begin{corollary}
Given two groups $H\substr G$, $\Lin(G,H)$ is solvable in polynomial time by $\AIP$ if and only if there is an Abelian endomorphism of $G$ that extends $\id_H$. If no such endomorphism exists, then $\Lin(G,H)$ is \NP-hard.
\end{corollary}

\begin{example}
Let $G$ be the dihedral group on four elements, and $H$ be the symmetric group on four elements.
Observe that $G$ can be seen as a subgroup of $H$
in a natural way: $H$ consists of all permutations on four elements, while $G$ contains only those that are symmetries of the square. The group $G$ is generated by the right $90$-degree rotation $r$ and an arbitrary reflection $f$ that leaves no element fixed. We consider two group homomorphisms $\varphi_1,\varphi_2$ with $\dom(\varphi_i)\substr G$ and $\im(\varphi_i)\substr H$. The domain of both homomorphisms is the subgroup $\{e, r, r^2, r^3\}\substr G$. Then, $\varphi_1$ is given by $r\mapsto r^2$, and $\varphi_2$ is given by $r\mapsto r$. The following hold:
\begin{itemize}
    \item %
    $\PLin(G,H,\varphi_1)$ is tractable, and solvable by AIP. However both
    $\Lin(G,$ $\dom(\varphi_1))$ and $\Lin(H,\im(\varphi_1))$ are \NP-hard.
    \item 
    $\PLin(G,H,\varphi_2)$ is \NP-hard. 
\end{itemize}
To see the first item, observe that the group homomorphism $\psi:G\rightarrow H$
  given by $r\mapsto r^2$ and $f\mapsto f$ is Abelian (its image is isomorphic
  to the direct product $\mathbb{Z}_2\times \mathbb{Z}_2$) and extends $\varphi_1$. Hardness of $\Lin(G,\dom(\varphi_1))$ is a consequence of the fact that the commutator subgroup of $G$ is $\{e,r^2\}$, so
$r^2\in \dom(\varphi_1)$ is included in the kernel of any Abelian endomorphism of $G$. Similarly, hardness of $\Lin(H,\im(\varphi_1))$
follows from the fact that the commutator subgroup of $H$ is the alternating group on four elements, and has $\im(\varphi_1)$ as a subgroup. \par
The second item can be proved by observing that the only normal subgroup of $G$ that does not intersect $\dom(\varphi_2)$ is the trivial subgroup, so any homomorphism $\psi:G \rightarrow H$ that extends $\varphi_2$ needs to be injective, and thus non-Abelian. 
\end{example}

We say that $\pcsp(\A,\B)$ is \emph{finitely tractable} if there is $\C$ such
that $\A\to\C\to\B$ and $\csp(\C)$ is solvable in polynomial time. The tractable
cases in~\Cref{th:promisemonoids} are in fact finitely tractable, as the next
result shows.

\begin{lemma}\label{lem:finitelytract}
Assume that $\PLin(M_1, M_2,\varphi)$ is in the positive part of~\Cref{th:promisemonoids}; i.e., there is an Abelian homomorphism $\psi:M_1\to M_2$ extending $\varphi$ and $\im(\psi)$ is a union of subgroups.
Then, $\PLin(M_1, M_2, \varphi)$ is finitely tractable. 
\end{lemma}
\begin{proof}
   By~\Cref{th:promisemonoids} there must be some Abelian homomorphism $\psi:M_1\rightarrow M_2$  extending $\varphi$ and $\im(\psi)$ is a union of subgroups.
    Let $M\substr M_2$ be the submonoid $\im(\psi)$. By assumption $M$ is
    Abelian and a union of subgroups. Let $N\substr M$ be the submonoid
    $\im(\varphi)$. We claim that $\Lin(M, N)$ is solvable in polynomial time.
    Indeed, consider the map $\id_M$. This map is an Abelian endomorphism of
    $M$, whose image is a union of subgroups. Moreover, $\id_M$ extends $\id_N$.
    So, by~\Cref{th:promisemonoids}, $\Lin(M,N)$ is solvable in polynomial time
    by $\BLP+\AIP$. \par
    The idea now is that $\Lin(M,N)$ can be ``sandwiched'' by the template
    $(\A,\B)$ (defined in~\Cref{def:promise_eqns})
    of $\PLin(M_1, M_2, \varphi)$. To make this formal, we need to produce a template 
    $\C$ in the same signature as $\A$ and $\B$ such that $\csp(\C)$ is $\Lin(M,N)$ up to relabeling some relations. The set $C$ equals $M$. The relation $R_{\times}^\C$
    consists of the triples $(s_1,s_2,s_3)\in M^3$ such that $s_1s_2=s_3$.
    Finally, for each $s\in \dom(\varphi)$, we define $R_{s}^\C=\{ \varphi(s)
    \}$. By construction, it holds that
    $\A \mapsto \C \mapsto \B$: the map $\psi$ is a homomorphism from $\A$ to $\C$, and the inclusion map is a homomorphism from $\C$ to $\B$. On the other hand, $\csp(\C)$ is easily seen to be equivalent to $\Lin(\C)$. Indeed, we can obtain $\Lin(\C)$ from $\C$ by relabeling each relation $R_s$ to $R_{\varphi(s)}$ and removing duplicate relations, which does not change the complexity of the related CSP.
\end{proof}

The power of $\BLP+\AIP$ is necessary in~\Cref{th:promisemonoids} in the sense
that $\AIP$ does not suffice for all monoids, even for (non-promise) CSPs,
unlike in the case of groups. Indeed, adding a fresh element to a group that
serves as the monoid identity fools $\AIP$.

\begin{lemma}\label{lem:blpaipneeded}
Let $G$ be an arbitrary Abelian group. Let $M$ be the monoid resulting from adding to $G$ a fresh element $e$ that serves as the monoid identity. Then $\Lin(M, M)$ is solvable by $\BLP+\AIP$ but not by $\AIP$.
\end{lemma}

\begin{proof}
    The fact that $\BLP+\AIP$ solves the $\Lin(M,M)$ follows by~\Cref{th:promisemonoids} from the fact that $M$ is Abelian and a union of subgroups. To rule out $\AIP$, we show that $\Lin(M,M)$ has no alternating polymorphisms; this suffices by~\Cref{thm:aip}. We begin with the following observation.
    Let $p: M^n \rightarrow M$ be a polymorphism whose $i$-th coordinate is inessential. Consider the  homomorphism $\tau_{p,i}$ that sends each element $s\in M$ to $p(e,\dots, s, \dots, e)$, where all arguments are equal to $e$ except for the $i$-th one, which is equal to $s$. Then it must be that $\tau_{p,i}$ is constant and equal to $e$. 
    Now suppose that $p$ is a $(2n+1)$-ary alternating polymorphism. Then define $q(x_1,\dots, x_{2n})=p(x_1,x_1,x_2,x_3,\dots,x_{2n-1})$. As $p$ is alternating, the first coordinate is inessential in $q$.
    By our previous observation, 
    $\tau_{q,1}$ is constant and equal to $e$. By definition,
    \[\tau_{q,1}(s)= q(s,e,\dots, e) = p(s,s,e,\dots, e)= p(s,e,\dots,e)p(e,s,e,\dots,e)=
    \tau_{p,1}(s)\tau_{p,2}(s).\]
    Hence $\tau_{p,1}(s) \tau_{p,2}(s)=e$ for all $s\in M$. 
    The only way that the product of two elements equals $e$ in $M$
    is that both elements are equal to $e$. Thus, both $\tau_{p,1}$
    and $\tau_{p,2}$ are constant and equal to $e$. This means that the first and the second coordinate are inessential in $p$.    
    However, as $p$ is
    alternating, $p$ is preserved under parity preserving permutations of its arguments, so the fact that its first and second coordinates are inessential means that in fact all its coordinates are inessential.
    However, if all coordinates of $p$ are inessential, then $p$ is constant, but this contradicts the fact that $p$ must be idempotent, as singleton
    unary relations are in $\Lin(M,M)$ and thus preserved by $p$.
\end{proof}

\mypar{Promise equations over semigroups}
As our next result,
we prove that every PCSP is polynomial-time equivalent to a problem of the form $\PLin(S_1, S_2, \varphi)$ over some semigroups $S_1,S_2$. 
Hence, extending our classification of promise equations beyond monoids is difficult in the sense that understanding the computational complexity of promise equations over semigroups is as hard as classifying all PCSPs. This result is analogous to the one known in the non-promise setting obtained in~\cite{KTT07:tcs}, whose proof we closely follow. One difficulty in lifting the result from~\cite{KTT07:tcs} is the lack of constants in the promise setting. 
The details can be found in~\Cref{sec:semigroups}.

\begin{theorem}\label{th:semigroups}
    Let $(\A,\B)$ be a template. Then there are semigroups $S_1,S_2$ and a semigroup homomorphism $\varphi$ with $\dom(\varphi)\substr S_1$ and $\im(\varphi) \substr S_2$ such that 
    $\pcsp(\A,\B)$ is polynomial-time equivalent to $\PLin(S_1, S_2, \varphi)$. 
\end{theorem}

\mypar{Monoidal minions}
As our third result, we investigate minions based on monoids. For PCSPs whose polymorphism minions are homomorphically equivalent to such minions, we establish a dichotomy. 
This is a building block in the proof of our main result, but may be interesting in its own right. In this direction, we show that for each monoidal minion $\Mscr$, there are PCSP templates whose polymorphism minions are isomorphic to $\Mscr$. 
For a finite set $[n]$, a tuple $(a_i)_{i\in [n]} \in M^n$ is called commutative if each pair of its elements commute. 

\begin{definition}
\label{def:monoidal_minion}
Given an element $a\in M$ of a monoid $M$, the \emph{monoidal minion}
$\Mscr_{M,a}$ is the one where for each $n\in\NN$ the elements
$\bm{b} \in \Mscr_{M,a}(n)$
are commutative tuples $\bm{b}\in M^n$
with $\prod_{i\in [n]} b_i=a$, and where for each $m\geq 1$ and each $\pi: [n]\rightarrow [m]$ the minor $\bm{b}^{(\pi)}$ is the tuple $\bm{c}\in M^m$ given by
$c_j= \prod_{i\in \pi^{-1}(j)} b_i$, and the empty product equals the identity element $e$. 
\end{definition}

\begin{theorem}
\label{th:main}
    Let $M$ be a finite monoid and let $a\in M$. Consider a template $(\A, \B)$ with
    $\pol(\A,\B)$ homomorphically equivalent to $\Mscr_{M,a}$. Then
    $\pcsp(\A,\B)$ is solvable in polynomial time by $\BLP+\AIP$ if and only if $a$ is regular in $M$.
    If $a$ is not regular, then $\pcsp(\A,\B)$ is \NP-hard. 
\end{theorem}

Next, we show that there are templates whose polymorphism minions are of the considered type (up to isomorphism).

\begin{theorem}
\label{th:examples}
    Let $M$ be a monoid, and $a\in M$ an arbitrary element. Then the template $(\A, \B)$ described below satisfies that $\pol(\A, \B)\simeq \Mscr_{M,a}$.\footnote{We use $\simeq$ to denote the isomorphism relation, i.e., the existence of a bijection between the minions that preserves arities and minor operations.} 
    The signature $\sigma$ of $\A$ and $\B$ contains three relation symbols:  a ternary symbol $R$, and two unary ones $C_0,C_1$. 
    We define $A=\{0,1\}$, 
    $R^\A=\{ (1,0,0), (0,1,0), (0,0,1)\}$, $C_0^\A=\{0\}$ and $C_1^\A=\{1\}$. The universe $B$ of $\B$ is $\Mscr_{M,a}(2)$. We define $R^\B$ as the set of triples in $\left(\Mscr_{M,a}(2)\right)^3$
    of the form $((c_1, c_2c_3), (c_2,c_1c_3), (c_3, c_1c_2))$, where $c_1,c_2,c_3\in M$  commute pairwise, and $c_1c_2c_3=a$.
    Finally, the unary relations $C_0^\B$ and $C_1^\B$ are the singleton sets containing the tuples $(e,a)$ and $(a,e)$ respectively.\footnote{The map $f: A \rightarrow B$ given by $0 \mapsto (e,a)$ and $1 \mapsto (a,e)$ is a homomorphism from $\A$ to $\B$. The structure $\A$ corresponds to the ``$1$-in-$3$'' template, where both constants are added, and $\B$ is the so-called ``free structure''~\cite{BBKO21} of $\Mscr_{M,a}$ generated by $\A$.}
\end{theorem}

Finally, we remark that monoidal minions are natural objects of study, as they include other relevant previously studied minions. 

\begin{remark}\label{rem:minions}
Consider the Abelian monoid $M=\{0,1,\epsilon\}$, whose multiplicative identity is $0$, and where $1\cdot 1= 1\cdot \epsilon = \epsilon \cdot \epsilon = \epsilon$.
The elements of $\Mcal_{M,1}$ are tuples with all zero entries except for a single $1$ entry. Hence $\Mcal_{M,1}$, corresponds to the so-called trivial minion $\mathscr{T}$
consisting of all projections (also known as dictators) on a two-element set.
  This minion represents the hardness boundary for CSPs, in the sense that a CSP
  is \NP-hard if and only if its polymorphism minion maps homomorphically to $\mathscr{T}$~\cite{Bulatov05:classifying,Zhuk20:jacm}.
\par
Another example of a monoidal minion is the one capturing the power of arc consistency from~\cite{Dalmau24:lics}. In fact, every \emph{linear} minion (in the sense of~\cite{cz23soda:minions}) is a union of monoidal minions.\footnote{We thank Lorenzo Ciardo for this observation.}
\par
If we allow infinite monoids to be considered, then monoidal minions include important minions that capture solvability via relevant algorithms. Consider the monoid
$M=\{ (r,z)\in \mathbb{Q} \times \mathbb{Z} \mid r\in [0,1], \text{ and } r=0 \text{ implies } z=0 \}$, where the binary operation is given by coordinate-wise addition, and the identity is $(0,0)$. Then $\Mscr_{M,(1,1)}$ is precisely the minion $\Mscr_{\BLP+\AIP}$
described in~\cite{BGWZ20}, which expresses the power of $\BLP+\AIP$. Similarly, the minions described in~\cite{BBKO21} to capture the power of $\BLP$ and $\AIP$ are monoidal minions as well. 
\end{remark}

\section{Monoidal Minions: Proof of~\Cref{th:main}}
\label{sec:monoidal}

\textbf{Solvability by $\BLP+\AIP$}\quad 
We show both directions. First we prove that $a$ being regular implies that
$\pcsp(\A, \B)$ is solvable by $\BLP+ \AIP$. We use the characterisation of the
power of $\BLP+\AIP$ from~\Cref{thm:blpaip} for the tractability part
of~\Cref{th:main}. Observe that if there is a minion homomorphism $\xi:
\Mscr_{M,a} \rightarrow \pol(\A,\B)$ and $p\in \Mscr_{M,a}$ is a $(2i+1)$-ary
$2$-block-symmetric element, then so is $\xi(p)$. Hence, showing that
$\Mscr_{M,a}$ has $2$-block-symmetric elements of all odd arities proves that
$\pcsp(\A,\B)$ is solvable in polynomial time by $\BLP+\AIP$. 
By~\Cref{lem:regularity}\,(2), $a^j=a$ for some $j>1$. Let
$b=a^{j-2}$, where $a^0=e$. Then $a^2b=a$, and $ab=ba$. For each $i\geq 1$ consider
the $(2i+1)$-ary element of $\Mscr_{M,a}$ consisting of $i+1$ consecutive $a$'s
followed by $i$ consecutive $b$'s.
To see that this this is indeed an element of $\Mscr_{M,a}$, observe that $a$
and $b$ commute, and $a^{i+1}b^i=a$ follows from $a^2b=a$. This tuple is
$2$-block-symmetric, with the blocks corresponding to $a$'s and $b$'s (of sizes
$i+1$ and $i$, respectively). 
\par Now we prove that if $\pcsp(\A,\B)$ is
solvable by \BLP+\AIP, then $a$ must be regular. If $\pcsp(\A,\B)$, then
$\pol(\A,\B)$ has a $(2i+1)$-ary $2$-block symmetric polymorphism $p_i$ for each
$i\geq 1$. As $\pol(\A,\B)$ is homomorphically equivalent to $\Mscr_{M,a}$, we
conclude there must be a $(2i + 1)$-ary $2$-block symmetric element $\bm{c_i}$
for each $i\geq 1$. Observe that $\bm{c_i}$ is a tuple in $M^{2i+1}$ of the form
$(\alpha_i, \dots, \alpha_i, \beta_i,\dots, \beta_i)$, where the first $i+1$
elements are equal to some $\alpha_i\in M$, and the last $i$ elements are equal
to some $\beta_i\in M$. By the definition of $\Mscr_{M,a}$ it must hold that
$\alpha_i^{i+1}\beta_i^i=a$, and that $\alpha_i\beta_i=\beta_i\alpha_i$. As $M$
is finite, there must be a pair $(\alpha, \beta)\in M^2$ that appears infinitely
often in the sequence $(\alpha_i,\beta_i)_{i\geq 1}$. Then, there must be two
indices $j> 2i+1$, with $i\geq 1$, satisfying $(\alpha,\beta)=(\alpha_i,
\beta_i)=(\alpha_j, \beta_j)$. The following chain of identities holds \[ a =
\alpha^{i+1}\beta^i = \alpha^{j+1}\beta^j = \left(\alpha^{i+1}\beta^j\right)^2
\left(\alpha^{j -2i -1}\beta^{j-2i} \right)= a^2 \left(\alpha^{j -2i
-1}\beta^{j-2i} \right). \] This shows that $a$ is regular. 

\smallskip
\noindent\textbf{\NP-hardness}\quad 
We prove the intractability part of~\Cref{th:main} (as well as other hardness results later in this paper) using the following result.

\begin{theorem}[\cite{BBKO21}]
    \label{th:hard_aux_general}
    Let $\Mscr= \pol(\A,\B)$, and let $K,L\geq 1$ be any fixed integers. Suppose that $\Mscr$ satisfies the following conditions:
    \begin{enumerate}
     \item $\Mscr= \bigcup_{\ell\in [L]} \Mscr_\ell$;
        
     \item for each $\ell\in [L]$, 
    there is a map $p \mapsto \mathcal{I}_\ell(p)$ that sends each $p\in \Mscr_\ell$ to a set of its coordinates $\mathcal{I}_\ell(p)$ of size at most $K$;
    
    \item for each $\ell\in [L]$ and  for each minor $p^{(\pi)}=q$, where $p,q\in \Mscr_\ell$, 
 $\pi(\mathcal{I}_\ell(p))\bigcap \mathcal{I}_\ell(q) \neq \emptyset$.
    \end{enumerate}
    Then $\pcsp(\A, \B)$ is \NP-complete. 
\end{theorem}

Given a template $(\A,\B)$, if there is a minion homomorphism
$\xi:\pol(\A,\B) \rightarrow \Mscr_{M,a}$ and $\Mscr_{M,a}$ satisfies the
conditions in~\Cref{th:hard_aux_general}, so does $\pol(\A,\B)$. Indeed, if $\Mscr_{M,a}=\bigcup_{\ell\in [L]} \Mscr_\ell$, then we can write $\pol(\A,\B)= \bigcup_{\ell \in [L]} \xi^{-1}(\Mscr_\ell)$.
Additionally, if the map $\mathcal{I}_\ell$ witnesses the condition in the theorem for
$\Mscr_\ell$, then the map $\mathcal{I}^\prime_\ell$ given by $p\mapsto \mathcal{I}_\ell(\xi(p))$ witnesses the same condition for $\xi^{-1}(\Mscr_\ell)$. Hence,
we show the hardness part of~\Cref{th:main} by proving that $\Mscr_{M,a}$ satisfies the assumptions in~\Cref{th:hard_aux_general} when $a\in M$ is not regular. \par
 
For a monoid $M$, we
we write $a \lneqsg b$ when $a\leqsg b$ holds but $b\leqsg a$ does not. We  use the following simple observation.

\begin{observation}
\label{obs:monoid}
Let $M$ be a monoid and $a,b,c\in M$ three elements that commute pairwise. Suppose that $abc \lneqsg ab$. Then $ac \lneqsg a$. 
\end{observation} 
\begin{proof}
    We prove the contrapositive. Suppose that $a\leqsg ac$. That is, there is some $d\in M$ that satisfies $acd=a$. 
    We have
    $(abc)d = (bac)d=b(acd)=ba=ab$, proving that $ ab \leqsg abc$.   
\end{proof}

    Assume that $a$ is not regular. That is, that $a^2 b \neq a$ for every $b\in M$. 
    Let $\bm{b}\in \Mscr_{M,a}(n)$ for some number $n\geq 1$. Using the fact that the elements $b_i$
    commute pairwise one can deduce that $\prod_{i\in I} b_i \leqsg \prod_{j\in J} b_j$ for all $J\subseteq I \subseteq [n]$. A coordinate $j\in
    [n]$ is called \emph{relevant} in $\bm{b}$ if $a
    \lneqsg \prod_{i\in [n]\setminus \{j\}} b_i$. 
     Consider the map $\mathcal{I}$ that assigns to each $\bm{b}\in \Mscr_{M,a}$ its set of relevant coordinates. Claims 1 through 3 proved below establish the required assumptions in~\Cref{th:hard_aux_general} with $L=1$ and $K=|M|$,
     thus showing \NP-hardness of
     $\pcsp(\A, \B)$. Throughout the proof we adopt the convention that empty products over a monoid equal the identity element. 
 \begin{description}
    \item[Claim 1: $\bm{b}$ has at most $|M|$ relevant coordinates.] Let $\{ i_1,\dots, i_h \} \subseteq [n]$ be the set of relevant coordinates of $\bm{b}$. Given $k\in [h]$ we define
    \[
    c_k = \prod_{j\in [k-1]} b_{i_j}, \quad \text{and} \quad
    d_k = \prod_{j\in [n] \setminus \{ i_1, \dots, i_k \}} b_j.
    \]
    The following hold: (1) $a= d_k c_k b_{i_k}$, (2) $b_{i_k}, c_k$ and $d_k$ commute pairwise, and (3) as $i_k$ is a relevant coordinate, it holds that $d_k c_k b_{i_k} \lneqsg d_k c_k$. Applying~\Cref{obs:monoid}, we obtain that
    $c_k b_{i_k} \lneqsg c_k$. Expanding the definition of $c_k$ this means that
    \[
    \prod_{j\in [k]} b_{i_j} 
    \lneqsg
    \prod_{j\in [k-1]} b_{i_j}.
    \]
    This holds for all $k\in [h]$,
    so in particular the products
    $\prod_{j\in [k]} b_{i_j}$ must be pairwise different and the number $h$ of relevant coordinates is at most $|M|$, proving the claim. 
  ~\\

   \item[Claim 2: Minors preserve relevant coordinates.] Let $\bm{c}=\bm{b}^{(\pi)}$, where $\pi:[n]\rightarrow [m]$ is a map and let $i\in [n]$ be a relevant coordinate of $\bm{b}$. We want to show that $j=\pi(i)$
    is a relevant coordinate of $\bm{c}$. Indeed, if that were not the case,
    using the equality
    $ \prod_{k\in [n]\setminus \pi^{-1}(j)} b_k
    = 
    \prod_{\ell \in [m]\setminus \{j\}} c_\ell
    $,
    we would have that
    \[
    \prod_{k\in [n]\setminus \pi^{-1}(j)} b_k \leqsg a.
    \]
     Using this together with the fact that
    $\prod_{k\in [n]\setminus \{i\}} b_k \leqsg \prod_{k\in [n]\setminus
     \pi^{-1}(j)} b_k$, where $i \in \pi^{-1}(j)$, 
     shows that
    \[
    \prod_{k\in [n]\setminus \{i\} } b_k \leqsg a,
    \]
    thus contradicting the fact that $i$ was a relevant coordinate of $\bm{b}$.
    ~\\
     \item[Claim 3: $\bm{b}$ has at least one relevant coordinate.] Suppose otherwise for the sake of contradiction. Then 
    for each $i\in [n]$ there is an element $c_i\in M$ such that $a c_i = \prod_{i\in [n] \setminus \{j\}} b_i $. Let $c=\prod_{i\in [n]} c_i$. One can check that that $a^2c=a$, contradicting our assumption that $a$ was not regular. Indeed, 
    \begin{align*}
        a^2 c & =  
        \left( \prod_{i = 1}^n b_i 
        \right) 
        \left( a c_1 \right)
        \left( \prod_{i=2}^n c_i \right) =  
        \left( \prod_{i= 1}^n b_i 
        \right) 
        \left( \prod_{i\in [n] \setminus \{1\} } b_i  \right)
        \left( \prod_{i=2}^n c_i \right)
        \\ 
         & =   
         \left( \prod_{i = 2}^n b_i 
        \right) 
        \left( a c_2 \right)
        \left( \prod_{i=3}^n c_i \right) =  
        \left( \prod_{i= 2}^n b_i 
        \right) 
        \left( \prod_{i \in [n] \setminus \{2\} } b_i  \right)
        \left( \prod_{i=3}^n c_i \right) \\ 
         & =   
         \left( \prod_{i = 3}^n b_i 
        \right) 
        \left( a c_3 \right)
        \left( \prod_{i=4}^n c_i \right) =  
        \left( \prod_{i= 3}^n b_i 
        \right) 
        \left( \prod_{i \in [n] \setminus \{3\} } b_i  \right)
        \left( \prod_{i=4}^n c_i \right) \\
         & = \dots = \left( \prod_{i= n}^n b_i 
        \right) 
        \left( \prod_{i \in [n] \setminus \{ n \} } b_i  \right) = a.
    \end{align*}
    Here we have repeatedly used the fact that the elements $b_i$ commute pairwise and in particular they commute with $a=\prod_{i\in [n]} b_i$.
\end{description}

\section{Equations Over Monoids and Groups: Proofs of~\Cref{th:promisemonoids} and~\Cref{th:promisegroups}}
\label{sec:equations}

We begin with a 
simple characterisation of the polymorphisms of promise equation templates.

\begin{lemma}\label{lem:polpromise}
    Consider a template $\PLin(Z_1, Z_2, \varphi)$ of promise equations over semigroups, monoids, or groups, respectively.
    A map $p: Z_1^n \rightarrow Z_2$ is a polymorphism of $\PLin(Z_1,Z_2,\varphi)$ if and only if $p$ is a semigroup, monoid, or group homomorphism, respectively, and $p(s,s,\dots, s)= \varphi(s)$
    for all $s\in \dom(\varphi)$.
\end{lemma}

\begin{proof}
We show the statement for the semigroup, monoid and the group case. The semigroup case is straightforward: $p$ is a polymorphism if and only if it preserves $R_\times$ and $R_{s}$ for all $s\in \dom(\varphi)$. Preserving $R_\times$ is equivalent to preserving the product operation from $Z_1^n$ to $Z_2$, and preserving $R_{s}$ means that $p(s,\dots, s)= \varphi(s)$. \par
The monoid case follows in the same way. Using the same reasoning we obtain that $p$
preserves the product operation from $Z_1^n$ to $Z_2$, and that 
$p(s,\dots,s)=\varphi(s)$ for all $s\in \dom(\varphi)$.
The only additional requirement is that $p(e_{Z_1},\dots, e_{Z_1}) = e_{Z_2}$. This follows from the facts that $e_{Z_1}\in \dom(\varphi)$, and $\varphi$ is a monoid homomorphism, so it must preserve identity elements. This means that $p(e_{Z_1},\dots, e_{Z_1})=
\varphi(e_{Z_1})= e_{Z_2}$.
\par
Finally, the group case is shown as the monoid case using that preserving inverse elements is just a consequence of preserving the product operation and preserving identity elements. 
\end{proof}

Let us discuss some key properties of polymorphisms that will be used in the proof of~\Cref{th:main}. Given an $n$-ary polymorphism $p$ of $\PLin(M_1,M_2, \varphi)$, we define $\mathcal{N}(p)$
as the submonoid $\{ p(s,\dots, s) \mid s\in M_1\}\substr M_2$. Given $i\in [n]$, we also define the submonoid $\mathcal{N}(p,i)\substr M_2$ as  
\[
    \{ p(s_1, \dots, s_n) \mid s_i\in M_1, \text{ and } 
    s_j=e \text{ when $j\neq i$} \}.
\]
We give some facts about these submonoids that follow directly from the definitions.
\begin{observation}
\label{obs:submonoid}
    Let $M_1,M_2$ be monoids and $\varphi$ a monoid homomorphism with $\dom(\varphi)\substr M_1, \im(\varphi) \substr M_2$. 
    Let $p$ be a $n$-ary polymorphism of $\PLin(M_1, M_2,\varphi)$. Then the following statements hold:
    \begin{enumerate}
    \item \label{it:submonoids1}
    The map
    $\phi: \prod_{i\in [n]} \mathcal{N}(p,i) \rightarrow M_2$ given by 
    $(s_1,\dots s_n)\mapsto \prod_{i\in [n]} s_i$ is a monoid homomorphism. In particular, 
    given $1\leq i < j \leq n$, 
    any two elements 
    $t_1 \in \mathcal{N}(p,i)$, $t_2\in \mathcal{N}(p,j)$ commute. \par
    \item \label{it:submonoids2} If $\mathcal{N}(p,i)= \mathcal{N}(p,j)$ for some $i\neq j \in [n]$ then $\mathcal{N}(p,i)$ is Abelian. \par
    \item \label{it:submonoids3} The submonoid $\mathcal{N}(p)$ is contained in $\im(\phi)$, where $\phi$ is as defined in~\Cref{it:submonoids1}. In particular, if $\mathcal{N}(p)$ is not Abelian, some $\mathcal{N}(p,i)$ must be non-Abelian. 
    \end{enumerate} 
\end{observation}

We are ready to prove our main result.

\promisemonoids*

\begin{proof}
    First we show that the existence of such homomorphism $\psi$ is equivalent
    to solvability by $\BLP+\AIP$. 
    We prove both implications. Suppose that such homomorphism $\psi$ exists. As $\im(\psi)$
    is a union of subgroups, by~\Cref{lem:regularity} there is some number $k>1$ such that
    $s^k=s$ for all $s\in \im(\psi)$. Let $n\geq 1$ be arbitrary. Consider the map $p: M_1^{2n+1} \rightarrow M_2$ given by 
    \[ (s_i)_{i\in [2n+1]}  \mapsto \left(\prod_{i\in [n+1]} \psi(s_i)\right)
    \left(
    \prod_{i\in [n]}  \psi(s_{i + n+ 1})^{k-2}
    \right),
    \]
    where the convention is that the zero-th power of an element equals the identity of the monoid. 
    We claim that $p$ is a $2$-block-symmetric polymorphism of $\PLin(M_1,M_2,\varphi)$
    with the first block consisting of the first $n+1$ coordinates, and the second block consisting of the rest. The fact that $p$ is a $2$-block-symmetric map with the blocks as claimed follows from the fact that $\psi$ is Abelian. To complete the argument, we show that $p$ is a polymorphism of $\PLin(M_1,M_2,\varphi)$ using the characterisation from~\Cref{lem:polpromise}. First, observe that the fact that $\psi$ is Abelian implies that $p$ is a monoid homomorphism. Indeed, 
    \begin{align*}
    p(s_1,\dots, s_{2n+1})& p(t_1,\dots,t_{2n+1})   \\ & =
   \left(\prod_{i\in [n+1]} \psi(s_i) \psi(t_i) \right)
    \left(
    \prod_{i\in [n]}  \psi(s_{i + n+ 1})^{k-1} \psi(t_{i + n+ 1})^{k-1}
    \right) \\ &  =  p(s_1t_1,\dots, s_{2n+1}t_{2n+1}),      
    \end{align*}
    so $p$ preserves products. Now we only need to prove that $p(s,\dots,s)=\varphi(s)$ for all $s\in \dom(\varphi)$ in order to show that $p$ is a polymorphism. To see that this holds, observe that 
    \[
    p(s,\dots,s) = \psi(s)^{n(k-1) + 1} = \psi(s) = \varphi(s), 
    \]
    where the last equality uses the fact that $\psi$ extends $\varphi$. This completes the proof of the first implication via~\Cref{thm:blpaip}. 
    \par
    In the other direction, suppose that $\PLin(M_1,M_2, \varphi)$ is solvable by $\BLP+\AIP$. That is, by~\Cref{thm:blpaip}, there is a $2$-block-symmetric polymorphisms $p_i$ of $\PLin(M_1,M_2, \phi)$ of arity $2i+1$ for each $i\geq 1$. For each $i$, we define three homomorphisms $\alpha_i,\beta_i, \gamma_i$ from $M_1$ to $M_2$. Given $a\in M_1$, we define 
    $\alpha_i(a)$ as the element $p_i(\bm{a})$, where $a_1=a$ and $a_j=e$ for all $j\neq 1$. Similarly, $\beta_i(a)$ is the element $p_i(\bm{a}^\prime)$, where $a_{i+2}^\prime=a$, and $a_{j}^\prime=e$ for all $j\neq i+2$. This way, given an arbitrary $\bm{b}\in M_1^{2i+1}$, it holds that
    \[
    p_i(\bm{b})=
    \left(\prod_{j=1}^{i+1} \alpha_i(b_j) 
    \right)
    \left( 
    \prod_{j=i+2}^{2i+1}
    \beta_i(b_j)
    \right).
    \] 
    Finally,
    given $a\in M_1$, the element $\gamma_i(a)$
    equals $p_i(a,a,\dots,a)$, so, $\gamma_i(a)=\alpha_i(a)^{i+1}\beta_i(a)^i$.    
    As the number of possible triples $(\alpha_i, \beta_i,\gamma_i)$ 
    is finite, there is a choice $(\alpha,\beta, \gamma)$ that appears infinitely often in the family $(\alpha_i, \beta_i, 
    \gamma_i)_{i=0}^\infty$.
    Let $i,j$ be such that 
    $(\alpha,\beta,\gamma)=(\alpha_i, \beta_i,\gamma_i)= (\alpha_j, \beta_j,\gamma_j)$
    and $j\geq  2i+1$. We claim that $\gamma$ has all the properties of the map $\Psi$ in the statement. We need to check that (I) $\gamma$ extends $\varphi$, (II) $\gamma$ has Abelian image, and (III) $\im(\gamma)$ is a union of subgroups. Property (I) follows from the fact that $p_i$ is a polymorphism. 
    To show property (II), we first need to make some observations about $\im(\alpha)$ and $\im(\beta)$. By definition, $\im(\alpha)=\mathcal{N}(p_j,1)$ and $\im(\beta)=\mathcal{N}(p_j,j+2)$.
    By~\Cref{it:submonoids1} in~\Cref{obs:submonoid}, this implies that $ab=ba$ for all $a\in \im(\alpha), b\in \im(\beta)$. 
    Using the fact that $p_j$ is $2$-block symmetric and $j\geq 2$, we can deduce that $\mathcal{N}(p_j,1)= \mathcal{N}(p_j,2)$
    and $\mathcal{N}(p_j,j+2)=\mathcal{N}(p_j,j+3)$.  By~\Cref{it:submonoids2} in~\Cref{obs:submonoid}, this implies that both $\im(\alpha)$ and $\im(\beta)$ are Abelian monoids. Having shown these properties of $\im(\alpha)$ and $\im(\beta)$ we are ready to show (II). We need to prove that $\gamma(a)\gamma(b)= \gamma(b)\gamma(a)$ for all $a, b\in M_1$. Indeed,
    \[
    \gamma(a)\gamma(b)= 
\left(    \alpha(a^{j+1})\beta(a^j)\right)
\left(
    \alpha(b^{j+1})\beta(b^j) \right) = 
    \gamma(b)\gamma(a),    
    \]
where the second equality uses the fact that all terms in the second expression commute due to our observations about $\im(\alpha)$ and $\im(\beta)$.
Finally, let us show that (III) holds. By~\Cref{lem:regularity}, we just need to show that for each $a\in M_1$ there is some $b\in M_2$ satisfying $\gamma(a)^2b= \gamma(a)$.
By our choice of $i$ and $j$, for all $a\in M_1$ it holds that    
\begin{align*}
 \gamma(a) & = 
\alpha(a^{i+1})\beta(a^i)=
\alpha(a^{j+1})\beta(a^j) \\ & = 
\alpha(a^{i+1})^2\alpha(a^{j - 2i - 1}) \beta(a^i)^2 \beta(a^{j-2i}) =
\gamma^2(a) \left(\alpha^{j - 2i -1}(a)\beta^{j - 2i}(a)\right),
\end{align*}
where the fourth equality uses that all the terms in the fourth expression commute by our observations about $\im(\alpha)$ and $\im(\beta)$. \par

Finally, let us prove the second part of the theorem. We show that $\PLin(M_1,M_2,\varphi)$ is \NP-hard assuming there is no Abelian homomorphism $\psi: M_1\rightarrow M_2$ extending $\varphi$ whose image is a union of subgroups. Let $\Mscr$ be the polymorphism minion of $\PLin(M_1,M_2,\varphi)$. Given a polymorphism $p\in \Mscr$, 
    we define $\mathcal{N}(p)$ as
    the submonoid $\{ p(s,\dots, s) \mid s\in M_1\}\substr M_2$. Observe that by assumption, for a given polymorphism $p$ it holds that  the monoid $\mathcal{N}(p)$ is non-Abelian or that $\mathcal{N}(p)$ is not a union of subgroups.  Define $\Omega$ 
    as the set of monoid homomorphisms $\psi: M_1 \rightarrow M_2$ for which $\im(\psi)$ is not a union of subgroups. By~\Cref{lem:regularity}, this happens precisely when $\im(\psi)$ contains some non-regular element $a\in M_2$.    
    Let $L=|\Omega| + 1$, and let 
    $K= \max( |M_2|, |\{ N\substr M_2 \mid N \text{ is non-Abelian } \}|) $.
    We use~\Cref{th:hard_aux_general} with the constants $L,K$ to show \NP-hardness. We define the following subminions of $\Mscr$. 
    \[
    \Mscr_{\Ab} = \{ 
    p \in \Mscr, \mid \mathcal{N}(p) \text{ is not Abelian}
    \},
    \]
    and given any monoid homomorphism $\psi \in \Omega$
     we set
    \[
    \Mscr_\psi = \{ 
    p \in \Mscr, \mid p(s,\dots, s) = \psi(s) \text{ for all } s\in M_1
    \}.
    \]
    By the previous observation it holds that
    \[
    \Mscr= \Mscr_{\Ab} \bigcup_{\psi \in \Omega} \Mscr_\psi.\]
    
    We give selection functions $\mathcal{I}$ for each of these sub-minions satisfying the assumptions of~\Cref{th:hard_aux_general}. Let $p$ be any $n$-ary polymorphism in $\Mscr_{\Ab}$. Given 
    $i\in [n]$ we define $\mathcal{N}(p,i)\substr M_2$ 
    as the submonoid 
    \[
    \{ p(s_1, \dots, s_n) \mid s_i\in M_1, \text{ and } 
    s_j=e \text{ when $j\neq i$} \}.
    \]
    We give some facts about these submonoids. \par
    
    \par

    Given an $n$-ary polymorphism $p\in \Mscr_{\Ab}$, we define $\mathcal{I}_{\Ab}(p)\subseteq [n]$ as the set of coordinates $i$ for which $\mathcal{N}(p,i)$ is non-Abelian.
    We claim that $\mathcal{I}_{\Ab}$ satisfies the assumptions of~\Cref{th:hard_aux_general}. Indeed, given some $n$-ary $p$:
    \begin{itemize}
        \item[\textbullet] $\mathcal{I}_{\Ab}(p)$ is non empty by~\Cref{it:submonoids3} in~\Cref{obs:submonoid}.
        \item[\textbullet] $|\mathcal{I}_{\Ab}(p)|\leq K$. Otherwise it would be that $\mathcal{N}(p,i)=\mathcal{N}(p,j)$ for some different $i,j\in \mathcal{I}_{\Ab}(p)$, contradicting the fact that $\mathcal{N}(p,i)$ is non-Abelian (by~\Cref{it:submonoids2} in~\Cref{obs:submonoid}).
        \item[\textbullet] Suppose that $p=q^{(\pi)}$ for some $m$-ary $q$
        and some $\pi:[m]\rightarrow [n]$. Let $i\in \mathcal{I}_{\Ab}(p)$, then 
        \[
        \mathcal{N}(p,i)\subseteq 
        \left\{ \prod_{j\in \pi^{-1}(i)} s_j \mid s_j\in \mathcal{N}(s,j) \text{ for all $j\in \pi^{-1}(i)$}
        \right\}.
        \]
        As $\mathcal{N}(p,i)$ is non-Abelian, some submonoid $\mathcal{N}(q,j)$ with $j\in \pi^{-1}(i)$ must be non-Abelian as well. This means that $\mathcal{I}_{\Ab}(p)\subseteq \pi(\mathcal{I}_{\Ab}(q))$.
    \end{itemize}
    
    Now consider an arbitrary homomorphism $\psi\in \Omega$ for which $\Mscr_\psi$ is non-empty. We define a selection function $\mathcal{I}_\psi$ satisfying the assumptions of~\Cref{th:hard_aux_general}. 
    Let $t\in \im(\psi)$ be a non-regular element, and let $s\in M_1$ be such that
    $\psi(s)=t$. Let $\Mscr_{M_2,t}$ be the monoidal minion defined in~\Cref{def:monoidal_minion}. Consider the map $\xi: \Mscr_\psi \rightarrow \Mscr_{M_2,t}$ that sends any $n$-ary polymorphism $p\in \Mscr_\psi$ to the tuple $(r_1,\dots, r_n)\in \Mscr_{M_2,t}(n)$ where for each $i\in [n]$
    \[
    r_i = p(s_1,\dots, s_n), \quad \text{ where } 
    s_i = s, \text{ and }
    s_j=e \text{ for all $j\neq i$ }.
    \]
    Observe that this is a well-defined minion homomorphism from $\Mscr_\psi$ to
    $\Mscr_{M_2,t}$. Indeed, first note that $(r_1,\dots, r_n)$ belongs to the
    second minion. This holds because $r_1r_2\dots r_n = p(s,\dots, s)=
    \psi(s)=t$, and, for each $i\in [n]$, the element $r_i$ belongs to $\mathcal{N}(p,i)$, so the $r_i$'s commute pairwise by~\Cref{it:submonoids1} in~\Cref{obs:submonoid}. One can also check that $\xi$ preserves minors.\par
    From the proof of~\Cref{th:main} 
    there is some selection function $\mathcal{I}$ on $\Mscr_{M_2,t}$
    satisfying the hypotheses of~\Cref{th:hard_aux_general} for some constant $K^\prime = |M_2| \leq K$ and $L=1$.
    Thus, we can define $\mathcal{I}_\psi$ on $\Mscr_\psi$ simply by setting $\mathcal{I}_\psi(p)= \mathcal{I}(\xi(p))$ for each polymorphism $p\in \Mscr_\psi$. \par
    Hence, have defined selection functions $\mathcal{I}_{\Ab}$ and $\mathcal{I}_\psi$ for each $\psi\in \Omega$ that satisfy the requirements of~\Cref{th:hard_aux_general}, showing that $\PLin(M_1, M_2, \varphi)$ is \NP-hard.
\end{proof}

\promisegropus*

\begin{proof}
    We prove both directions. 
    The hardness case follows from~\Cref{th:promisemonoids}. Indeed, $\PLin(G_1,G_2,\varphi)$
    is a template of promise equations over monoids (where the monoids just
    happen to be groups). Suppose that there is no Abelian group homomorphism
    $\psi: G_1 \rightarrow G_2$ that extends $\varphi$. Observe that a monoid
    homomorphism between two groups must also be a group homomorphism, so there
    is no Abelian monoid homomorphism $\psi: G_1 \rightarrow G_2$ that extends
    $\varphi$. Thus, by~\Cref{th:promisemonoids}, $\PLin(G_1,G_2,\varphi)$ is \NP-hard. \par

    In the other direction, suppose that such a $\psi$ exists. We show that
    $\PLin(G_1,G_2,\varphi)$ is solved by AIP using~\Cref{thm:aip}. Let $n$ be any odd arity and let $p: G_1^n \rightarrow G_2$ be the map given by
    $p(g_1,\dots, g_n) \mapsto \prod_{i\in [n]} t_i$, where $t_i=\psi(g_i)$ for every odd $i$, and $t_i=\psi(g_i)^{-1}$ for every even $i$. Then $p$ is an  alternating polymorphism of $\PLin(G_1,G_2,\varphi)$. 
\end{proof}

\section{Equations over Semigroups: Proof of~\Cref{th:semigroups}}
\label{sec:semigroups}

A \emph{digraph} $\D$ is a relational structure whose signature consists of a single binary relation symbol $E$. 

We follow closely the ideas from~\cite[Theorem 7]{KTT07:tcs}. That result states that every CSP is polynomial-time equivalent to a problem of the form $\Lin(S,S)$ for some semigroup $S$. Their proof uses the fact that every CSP is polynomial-time equivalent to another CSP whose template is a digraph $\D$ with all singleton unary relations~\cite{Feder98:monotone}. The fact that they consider these unary relations on $\D$ yields equations in $\Lin(S,S)$ where all constants are allowed. For PCSPs, however, this is our starting point.

\begin{theorem}[\cite{BG21:sicomp}]
\label{th:pcsps_to_digraphs}
For every  template $(\A_1, \A_2)$ there is a template $(\D_1, \D_2)$ of digraphs such that $\pcsp(\A_1, \A_2)$ is polynomial-time equivalent to $\pcsp(\D_1, \D_2)$.    
\end{theorem}

The fact that we lack singleton unary relations in the templates $(\D_1,\D_2)$ is the main obstacle for applying the techniques from~\cite{KTT07:tcs}. We overcome this by extending our digraphs with an additional edge joining two fresh distinguished vertices. 
The relational signature $\sigma^+$ contains one binary relation symbol $E$, and two unary relation symbols $P, Q$. Given a digraph $\D$, we write $\D^+$ for the $\sigma^+$ structure 
defined by $D^+ = D \cup \{p,q\}$, where $p$ and $q$ are fresh vertices, 
$E^{\D^+}=E^{\D} \cup \{(p,q)\}$, $P^{\D^+}=\{p\}$, and $Q^{\D^+}=\{q\}$.

\begin{lemma}
\label{lem:extended_digraphs}
Let $(\D_1, \D_2)$ be a  template of digraphs. Then $\pcsp(\D_1, \D_2)$ is polynomial-time Turing-equivalent to $\pcsp(\D_1^+, \D_2^+)$.    
\end{lemma}
\begin{proof}
    We give polynomial-time Turing reductions in both directions. First, we reduce from 
    $\pcsp(\D_1, \D_2)$ to $\pcsp(\D_1^+, \D_2^+)$. 
    We consider two cases. Suppose that $E^{\D_2}$ is empty. Then $\pcsp(\D_1, \D_2)$ amounts to deciding whether a given instance $\I$ has an edge or not, which takes polynomial time.
    Otherwise, assume that $E^{\D_2}$ is non-empty. Then our reduction takes any instance $\I$ of $\pcsp(\D_1, \D_2)$ and considers it as an instance of $\pcsp(\D_1^+, \D^+_2)$ where the unary relations are empty. Clearly, if $\I$ maps homomorphically to $\D_1$ then it also maps homomorphically to $\D_1^+$ using the same homomorphism. Otherwise, if $\I$ does not map homomorphically to $\D_2$ then it cannot map homomorphically to $\D_2^+$. Indeed, to see this observe that the digraph resulting from 
    of $\D_2^+$ (by forgetting about the $P,Q$ relations) maps homomorphically to $\D_2$:
    it suffices to send the edge $(p,q)$ to an arbitrary edge in $E^{\D_2}$, which is non-empty by assumption. 
    \par
    Now we describe a polynomial-time reduction from $\pcsp(\D_1^+, \D_2^+)$ to $\pcsp(\D_1, \D_2)$. The reduction considers an instance $\I$ of $\pcsp(\D_1^+, \D_2^+)$ and  checks in polynomial time whether every connected component of $\I$ that intersects $P^\I$ or $Q^\I$ maps homomorphically to the edge structure $\bm{W}$ with $W=\{p,q\}$, $E^{\bm{W}}=\{(p,q)\}$, $P^{\bm{W}}=\{p\}$, and $Q^{\bm{W}}= \{q\}$. If this is not the case, $\I$ is rejected. Otherwise, we remove from $\I$ the components that intersect $P^\I$ or $Q^\I$. Next, we check in polynomial time whether each remaining component of $\I$ can be mapped homomorphically to $\bm{W}$, and remove the ones that do. The resulting instance $\I^\prime$ is equivalent to the original $\I$ in the sense that $\I$ maps to $\D_i^+$ if and only if $\I^\prime$ does so as well. Furthermore, 
    observe that a homomorphism from $\I^\prime$ to $\D_i^+$ cannot include $p$ and $q$ in its image, as there are no components in $\I^\prime$ that map homomorphically to $\bm{W}$. This means that $\I^\prime$ maps to $\D_i^+$ if and only if it maps to $\D_i$. Hence, as the last step in our reduction we simply use $\I^\prime$ as an instance of $\pcsp(\D_1, \D_2)$.     
\end{proof}

A semigroup $S$ is a \emph{right-normal band} if $ss=s$ for all $s\in S$ 
and $rst=srt$ for all $r,s,t\in S$.
Recall that we write $s\sim r$ if $s\leqsg r$ and $r\leqsg s$ hold.
It follows from the definitions that the quotient $\widehat{S}=S/\sim$
inherits the semigroup structure from $S$. Moreover, $\widehat{S}$ is a
\emph{semilattice}, meaning that it is an Abelian semigroup where every element
is idempotent. Given an instance $\I$ of 
$\Lin(S,S)$ we denote by $\widehat{\I}$
the corresponding instance over $\widehat{S}$, where every constant $s$ is substituted by its $\sim$ class $\hat{s}$. 

We need two lemmas from~\cite{KTT07:tcs} and a simple observation.

\begin{lemma}[\cite{KTT07:tcs}]
\label{lem:semilattice}
Let $S$ be a semilattice. Then
$\Lin(S,S)$ can be solved in polynomial time. Moreover, if an instance $\I$ has a solution, it also has a unique minimal one (with respect to the $\leqsg$ preorder) that can be obtained in polynomial time. 
\end{lemma}

\begin{lemma}[\cite{KTT07:tcs}]
\label{lem:semilatice_to_band}
    Let $S$ be a right-normal band. Then an instance $\I$ of $\Lin(S,S)$ is
    solvable if it has a solution $f$  satisfying
    $f(x)\in \hat{s}_x$, for all $x\in I$, where the map $x\mapsto \hat{s}_x$ is the minimal solution of $\widehat{\I}$ in $\Lin(\widehat{S}, \widehat{S})$.
\end{lemma}

\begin{observation}
\label{obs:right_normal_band_identity}
    Let $S$ be a right-normal band, and let $s, s^\prime, t\in S$ be three arbitrary elements with $s\sim s^\prime$. Then $st=s^\prime t$.
\end{observation}
\begin{proof}
    As $s\sim s^\prime$, it must hold that $s=s^\prime r^\prime$ and $s^\prime= s r$ for some $r,r^\prime \in S$. 
    Thus,
    $st= s^\prime r^\prime t = 
    s r r^\prime t$, and 
    $s^\prime t = s r t = s^\prime r^\prime r t = s r r^\prime r t = 
    s r r^\prime t$,
    where the last equality holds since $S$ is a right-normal band.
\end{proof}

Let $\D$ be a digraph.
We define a semigroup $S_D$ related to $\D$ in a similar fashion as~\cite{KTT07:tcs}. The main difference is that we need to ``plant'' a special subsemigroup $S_W$ inside $S_D$ that is used later as the set of constants in our promise equations. 
The semigroup $S=S_D$ is a right-normal band. It has the following
$\sim$-classes: $V^{\Ls}, V^{\Rs}, V^{\LCs}, V^{\LRs}, V^{\CRs}, E^{\Cs}, 0$,
described as follows.
Given $\square\in \{
\Ls, \Rs, \LCs, \LRs, \CRs \}$ the class $V^\square$ is a copy of $D\cup \{p,q\}$. That is,
$V^\square = \{ v^\square \mid v\in D \} \cup \{ p^\square, q^\square \}$. The class $E^{\Cs}$ is a copy of $E^D\cup \{(p,q)\}$, meaning that
$E^{\Cs}= \{ 
(u,v)^{\Cs} \mid (u,v)\in E^D \} \cup \{ (p,q)^{\Cs} \}$.
The letters $\Ls,\Rs,$ and  $\Cs$ stand for \emph{left, right, and center}, respectively. 
Finally, the class $0$ contains a single element $0$. By~\Cref{obs:right_normal_band_identity}, in a right-normal band $T$ it must hold that $st= s^\prime t$ for all $s,s^\prime, t\in T$ with $s\sim s^\prime$. Hence, given a $\sim$-class $C\subseteq T$ and an element $t$ we abuse the notation and write $C t$ to denote the product of an arbitrary element from $C$ with $t$. The product operation in $S$ is given by the following rules:

\begin{align*}
V^{\Rs} v^{\Ls} = V^{\Ls} v^{\Rs} = &
V^{\LRs} v^{\Rs} = V^{\LRs} v^{\Ls}
= V^{\Ls} v^{\LRs} = 
V^{\Rs} v^{\LRs} = v^{\LRs} & \\
& V^{\Ls} v^{\LCs} = V^{\LCs} v^{\Ls} =
 E^{\Cs} v^{\Ls} = E^{\Cs} v^{\LCs} = 
v^{\LCs} & \\
& V^{\Rs} v^{\CRs} = V^{\CRs} v^{\Rs} =
 E^{\Cs} v^{\Rs} = E^{\Cs} v^{\CRs} = 
v^{\CRs}, & 
\end{align*}
where $v$ is an arbitrary element in $D\cup \{p,q\}$. Additionally,
\[
V^{\Ls} (u,v)^{\Cs}= 
V^{\LCs} (u,v)^{\Cs} = u^{\LCs}, \quad \text{and} \quad 
V^{\Rs} (u,v)^{\Cs} = 
V^{\CRs} (u,v)^{\Cs} = v^{\CRs},
\]
where $(u,v)$ belongs to $E^D \cup \{ (p,q) \}$. Finally, all other products not described above have $0$ as their result. \par
Let us give some intuition about the semigroup $S_D$. Our goal is to encode the incidence structure of the digraph $D$ in a semigroup. The construction $S_D$ achieves this as follows. The rule $V^{\Ls} (u,v)^{\Cs} = u^{\LCs}$ states that multiplying any
$\Ls$-copy of a vertex $w^{\Ls}$ with an edge 
$(u,v)^{\Cs}$ results in the $\LCs$-copy of the edge's source
$u^{\LCs}$. Similarly, we can extract information about an edge's target by
multiplying with an $\Rs$-copy of a vertex using the rule
$V^{\Rs} (u,v)^{\Cs} =  v^{\CRs}$. Finally, the identities
$V^{\Rs} v^{\Ls} = V^{\Ls} v^{\Rs} = v^{\LRs}$ show that an $\Ls$-copy of a
vertex $v^{\Ls}$ and an $\Rs$-copy of a vertex $u^{\Rs}$ commute if and only if they are copies of the same vertex (i.e., $v=u$). 
\par
We define the subsemigroup $S_W \substr S_D$ as the one containing the elements $0, (p,q)^{\Cs}, p^\square, q^\square$  for $\square\in \{
\Ls, \Rs, \LCs, \LRs, \CRs \}$. Observe that for any digraph $D$, the quotient
$\widehat{S_D}=S_D/ \sim$ is isomorphic to $\widehat{S_W}=S_W/\sim$. This
semilattice is depicted in~\Cref{fig:lattice}.

\begin{figure}[h]
\centering
\includegraphics[width=0.3\textwidth]{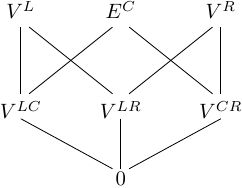}
\caption{The semilattice $\widehat{S_D}$, where lines indicate the order.}
\label{fig:lattice}
\end{figure}

\begin{lemma}
\label{lem:digraphs_to_eqns}
    There is a polynomial-time algorithm $\Phi$ that takes as an input a $\sigma^+$-structure $\I$
    and outputs a system of equations $\Phi(\I)$ with constants in $S_W$
    satisfying that, for any digraph $\D$, $\I$ maps into $\D^+$ if and only if $\Phi(\I)$ has a solution over $S_D$.
\end{lemma}
\begin{proof}
This follows the first reduction in~\cite[Theorem 7]{KTT07:tcs} while making sure that all constants remain in $S_W$.
We construct the system $\Phi(\I)$. 
For every vertex $v\in I$ we include variables $v^{\Ls},v^{\Rs}$. For each $\square\in \{
\Ls, \Rs\}$ we include the constraint $v^\square\in V^\square$, which is a shorthand for the equations $p^\square v^\square = v^\square$ and $v^\square p^\square = p^\square$.
We also include the equation $p^{\LRs} v^{\Ls} = p^{\LRs} v^{\Rs}$. If $v\in P^\I$ we include the constraints $v^\square = p^\square$ for $\square\in \{
\Ls, \Rs\}$. Similarly, if $v\in Q^\I$, then we include the constraints $v^\square = q^\square$. For each edge $(u,v)\in E^\I$ we include a variable $(u,v)^{\Cs}$ in $\Phi(\I)$, together with the constraint $(u,v)^{\Cs}\in E^{\Cs}$, which is a shorthand for
the equations
$(u,v)^{\Cs} (p,q)^{\Cs} = (p,q)^{\Cs}$ and
$(p,q)^{\Cs} (u,v)^{\Cs} = (u,v)^{\Cs}$.
Finally, we also add the equations
$p^{\LCs} (u,v)^{\Cs} = p^{\LCs} u^{\Ls}$
and 
$p^{\CRs} (u,v)^{\Cs} = p^{\CRs} v^{\Rs}$. 
The resulting system $\Phi(\I)$ satisfies the statement of the theorem. 
\end{proof}

\begin{lemma}
\label{lem:eqns_to_digraphs}
    There is a polynomial-time algorithm $\Psi$ that takes as an input a system of equations $\X$ with constants in $S_W$ and produces one of the following outcomes:
    \begin{enumerate}
    \item[(I)] It outputs a $\sigma^+$-structure $\Psi(\X)$ that maps into $\D^+$ for a digraph
    $\D$ if and only if $\X$ has a solution over $S_D$, or
    \item[(II)] it rejects $\X$ and $\X$ has no solution over $S_D$ for any digraph $\D$.  
    \end{enumerate}
\end{lemma}
\begin{proof}
    We describe the algorithm $\Psi$.
    This algorithm  is meant to transform the system $\X$ into a system of the form $\Phi(\I)$, for the algorithm $\Phi$ given in~\Cref{lem:digraphs_to_eqns} and some $\sigma^+$-structure $\I$. This time we follow the second reduction in~\cite[Theorem 7]{KTT07:tcs} while making sure that all constants in $\X$ remain in $S_W$ throughout all the transformations.    
    \par
    Without loss of generality, we may assume that every equation in $\X$ is initially of the form $x_1x_2=x_3$, for some variables $x_1, x_2, x_3$, or of the form $x=s$, for some variable $x$ and some element $s\in S_W$. 
    Consider the system $\widehat{\X}$ with constants in $\widehat{S_W}=S_W/\sim$.
    By~\Cref{lem:semilattice} we can find a minimal solution of $\widehat{\X}$ in polynomial time. 
    If such a solution does not exist, then the system $\X$ is not satisfiable over $S_D$ for any digraph $\D$, and the algorithm $\Psi$ just rejects it.  
    Otherwise, suppose that the system $\widehat{\X}$ has some minimal solution. This solution maps each variable $x\in X$ to a $\sim$-class $C_x$ of $S_W$. Consider an arbitrary digraph $\D$. Using the observation that $\widehat{S_W}\simeq \widehat{S_D}$ 
    and~\Cref{lem:semilatice_to_band}, we deduce that $\X$ has a solution over $S_D$ if and only if it has a solution where the value of each variable $x\in X$ belongs to the class $C_x$. Given a class $C_x$, we define the constant $c_x\in S_W$
    as 
    \begin{itemize}
        \item $p^\square$ if $C_x$ is the class $V^\square$ for $\square\in \{ 
         \Ls,\Rs, \LCs, \LRs, \CRs \}$,
        \item $(p,q)^{\Cs}$ if $C_x= E^{\Cs}$, or
        \item $0$ if $C_x= 0$.
    \end{itemize}
    For each variable $x\in X$ we add the equations $c_x x = x$ and $x c_x = c_x$. These equations are equivalent to the constraint that $x\in C_x$ (and we use $x\in C_x$ as a shorthand for those equations), so the resulting system is satisfiable over a semigroup $S_D$ if and only if the original one was. Additionally, once every variable $x$ is constrained to take values inside $C_x$, we can replace every equation of the form $x_1 x_2 = x_3$ in $\X$ with the equation
    $c_{x_3} x_2 = c_{x_3} x_3$ to yield an equivalent system. 
    Indeed, it must hold that $c_{x_i} x_i = x_i$, so the equation $x_1x_2=x_3$ is equivalent to $c_{x_1} x_1 c_{x_2} x_2 = c_{x_3} x_3$. Not only that, but $S_D$ is a normal band and $x_1c_{x_1}=c_{x_1}$, so last equation is equivalent to $c_{x_1}c_{x_2}x_2=c_{x_3}x_3$. Finally, the classes $C_{x_1}, C_{x_2}, C_{x_3}$ were part of a solution to $\widehat{\X}$, so
     it must be that $c_{x_1} c_{x_2} \sim c_{x_3}$, and by~\Cref{obs:right_normal_band_identity} it holds that $c_{x_1} c_{x_2} c_{x_1} = c_{x_3} c_{x_1}$. \par
     
     Every resulting equation of the form $0x_1=0x_2$ is trivially satisfied and can be discarded. Consider a variable $x\in X$ whose corresponding class $C_x$ is $0$. 
     As we have removed every equation of the form $0x_1=0x_2$, $x$ can only appear in constraints of the form $x \in 0$, and $x = 0$. These are trivially satisfiable by any assignment that maps $x$ to $0$, so we can remove the variable $x$ and all equations containing it. \par
     We are left with a system $\X$
    where each variable is bound to a class $V^\square$ for $\square\in \{ 
    \Ls, \Rs, \LCs, \LRs, \CRs \}$
    or $E^{\Cs}$. Consider a variable $x\in X$ bound to the class $V^{\LCs}$. Suppose this variable appears in some equation of the form $c_1 x = c_1 y$, and consider the class $C$ of $c_1$. By construction, it must be that $C\geqsg V^{\LCs}$ in $\widehat{S_W}$. However, we have removed all equations containing $0$, so the only possibility left is that $C = V^{\LCs}$. Suppose that we replace the requirement $x\in V^{\LCs}$ with $x\in V^{\Ls}$ and every equation of the form $x=v^{\LCs}$, where $v^{\LCs}\in S_W$
    is a constant, with $x= v^{\Ls}$. We claim the system $\X$ remains equivalent after these changes. Indeed, this results from the observation that $V^{\LCs} v^{\Ls}= V^{\LCs} v^{\LCs}$ in any semigroup $S_D$ for any vertex $v\in D^+$.
    By the same logic we can also replace any requirement of the kind $x\in V^{\LRs}$ or $x \in V^{\CRs}$ with $x\in V^{\Rs}$.\par
    Consider any equation of the form $x=(u,v)^{\Cs}$ for a constant $(u,v)^{\Cs}$. This equation is equivalent to the constraints $p^{\LCs} x = p^{\LCs} y$,
    $p^{\CRs} x = p^{\CRs} z$,
    $y=u^{\Ls}$ and $z= v^{\Rs}$, where 
    $y$ and $z$ are fresh variables, further restricted to 
    $y\in V^{\Ls}, z\in V^{\Rs}$.
    Hence, we can substitute in $X$ the original equation with these constraints to obtain an equivalent system. 
    \par
    
    Consider an equation of the form
    $c x = c y$, where both $x,y$ are constrained to be in $c$'s $\sim$-class. This equation holds if and only if $x=y$. Hence, we may remove this equation and identify both variables $x,y$ together.    
    \par
    This far we have obtained a system $\X$ where each variable is bound to either $V^{\Ls}, V^{\Rs}$ or $E^{\Cs}$, and the only constants are among $p^{\Ls}, p^{\Rs}, q^{\Ls}$, $q^{\Rs}$. After identifying variables and adding dummy variables if necessary we can assume the following:
    \begin{itemize}
        \item For each variable $x\in X$ constrained by $x\in E^{\Cs}$ there is exactly one variable $x_{\Ls}$ constrained by $x_{\Ls}\in V^{\Ls}$
        in an equation of the form $p^{\LCs} x = p^{\LCs} x_{\Ls}$, and exactly one variable $x_{\Rs}$ constrained by $x_{\Rs}\in V^{\Rs}$ that appears in an equation of the form $p^{\CRs} x = p^{\CRs} x_{\Rs}$.
        \item For any variable $x_{\Ls}$ constrained by
        $x_{\Ls}\in V^{\Ls}$, there is exactly one variable
        $x_{\Rs}$ constrained by $x_{\Rs}\in V^{\Rs}$ that appears in an equation of the form $p^{\LRs} x_{\Ls} = 
        p^{\LRs} x_{\Rs}$. The same remains true after swapping $\Rs$ and $\Ls$.   
        \item No two different variables $x,y\in X$ constrained by $x,y\in E^{\Cs}$ satisfy $x_{\Ls}= y_{\Ls}$ and $x_{\Rs} = y_{\Rs}$.
        \item Not considering equations that are part of the constraints $x\in C$
        for some $\sim$-class $C$, each equation is of the form (i) $p^{\LRs} x = p^{\LRs} y$ with $x\in V^{\Ls}$ and $y\in V^{\Rs}$,
        (ii) $p^{\LCs} x = p^{\LCs} x_{\Ls}$
        or  $p^{\CRs} x = p^{\CRs} x_{\Rs}$
        for some $x\in E^{\Cs}$, or
        (iii) $x=p^{\square}$ or $x=q^\square$ for
        $\square\in \{\Ls,\Rs\}$.
    \end{itemize}
    
One can see that such a system corresponds to $\Phi(\I)$ for some $\sigma^+$-structure $\I$ that can be built in polynomial time. Then $\Psi$ returns $\I$, which satisfies our requirements by~\Cref{lem:digraphs_to_eqns}.
\end{proof}

\begin{corollary}\label{cor:digraphs}
Let $(\D_1, \D_2)$ be a template of digraphs. Then $\pcsp(\D_1, \D_2)$ is polynomial-time Turing-equivalent to $\PLin(S_{D_1}, S_{D_2}, \varphi)$, where $\varphi=\id_{S_W}$.    
\end{corollary}
\begin{proof}
    We show that  $\PLin(S_{D_1}, S_{D_2}, \varphi)$ is polynomial-time equivalent to $\pcsp(\D_1^+, \D_2^+)$, which suffices by~\Cref{lem:extended_digraphs}. Observe that algorithm $\Phi$ given in~\Cref{lem:digraphs_to_eqns} is a polynomial-time Turing reduction from $\pcsp(\D_1^+, \D_2^+)$ to  $\PLin(S_{D_1}, S_{D_2}, \varphi)$, and algorithm $\Psi$, given in~\Cref{lem:eqns_to_digraphs} is a polynomial-time Turing reduction in the other direction. 
\end{proof}

\Cref{cor:digraphs} and~\Cref{th:pcsps_to_digraphs} establish~\Cref{th:semigroups}.

\section{Explicit Templates: Proof of~\Cref{th:examples}}
    We describe a bijective minion homomorphism $\xi: \pol(\A,\B) \rightarrow \Mscr_{M,a}$. Let us introduce some notation before the start. We identify the powerset $2^{[n]}$ with the set of tuples $\{0,1\}^n$ by associating each set $S\subseteq [n]$ to the $n$-tuple whose $i$-th entry is one if and only if $i\in S$ (i.e., the  characteristic vector of $S$). Thus, we see a $n$-ary polymorphism $p\in \pol(\A, \B)$ as a map from $2^{[n]}$ to
    $\Mscr_{M,a}(2)$. Following this convention, three sets $X_1, X_2,
    X_3\subseteq [n]$ belong to the relation $R^{\A^n}$ if and only if they are
    a partition of $[n]$. Similarly, the unary relation $C_0^{\A^n}$ contains
    only the empty set, and $C_1^{\A^n}$ contains only the whole set $[n]$. \par    
    Let us carry on with the description of $\xi$.
    Given a $n$-ary polymorphism $p\in \pol(\A,\B)$, $\xi$ maps it to the tuple $\bm{b}_p=(b_{p,1},\dots, b_{p,n}) \in \Mscr_{M,a}$ defined as follows. For each $i\in [n]$, let $(c_{i,1},c_{i,2})=p(\{i\})$. Then we set $b_{p,i}=c_{i,1}$.
    In order to prove that $\xi$ is a bijective minion homomorphism we need to show that (I) $\bm{b}_p$ is an element of $\Mscr_{M,a}$ 
    for all polymorphisms $p$, (II) that $\xi$ preserves minor operations, and (III) that $\xi$ is a bijection. \par
    Before moving on with the rest of the proof, we recall the definition of $R^\B$:
    \begin{align}
        \nonumber
        & ((r_1,s_1), (r_2, s_2), (r_3,s_3))\in R^\B \text{ if and only if } \\
        &
        \label{eq: aux_free_structure}
        r_1,r_2,r_3 \text{ commute pairwise, and }
        s_1= r_2r_3, \quad s_2 = r_3r_1, \quad s_3= r_1r_2.
    \end{align}
    Below we prove some claims about the map $\xi$ and polymorphisms of $(\A,\B)$
    that are used to show (I), (II), and (III).
    Fix some $n$-ary polymorphism $p$. 
    \begin{description}[style=unboxed]
     \item[Claim 1: $b_{p,i}$ and $b_{p,j}$ commute in $M$ for any $1\leq i<j\leq n$.]  
     Fix different indices $i,j\in [n]$. Clearly, the sets $\{i\}, \{j\}, [n]\setminus \{i,j\}$
     form a partition of $[n]$, so 
     $(p(\{i\}), p(\{j\}), p( [n]\setminus \{i,j\})$ must belong to $R^\B$. 
     Thus, there are elements $r_1,r_2,r_3\in M$
     witnessing~\Cref{eq: aux_free_structure}.
     In particular, $b_{p,i}=r_1$ and $b_{p,j}=r_2$, so $b_{p,i}$ and $b_{p,j}$ commute, establishing the claim. 
     \item[Claim 2:
     $p({[n]})=(a,e)$, and $p(\emptyset)=(e,a)$. In general, let $(c_1, c_2)= p(X)$ for some $X \subseteq {[n]}  $. Then $p( {[n]}\setminus X)=(c_2,c_1)$.] 
    The facts that $p([n])=(a,e)$, and $p(\emptyset)=(e,a)$ follow from the requirements that $p$ must preserve $C_1$ and $C_0$ respectively. Let us show the second part of the claim. 
    Observe that $(X, [n]\setminus X , \emptyset) \in R^{\A^n}$ so the image of this triple through $p$ belongs to $R^\B$. Hence, there are elements $r_1,r_2,r_3\in M$     witnessing~\Cref{eq: aux_free_structure}.  
    It also holds that $p(\emptyset)=(e,a)$, so 
    $r_3=e$, and $p(X)=(r_1,r_2)$, $p([n]\setminus X) = (r_2, r_1)$, as desired. 
    \item[Claim 3: Let $(b_1,b_2)=p(X)$
    $(c_1,c_2)=p(Y)$  $(d_1,d_2)=p(X\cup Y)$
    for two disjoint sets $X,Y\subseteq {[n]}$.
    Then $d_1=b_1c_1$.] Indeed, consider the set $Z=[n]\setminus (X\cup Y)$.
    By Claim~2 it holds that $p(Z)= (d_2, d_1)$. Moreover, it holds that $(X,Y,Z)\in R^{\A^n}$, so using the characterisation given in~\Cref{eq: aux_free_structure} we obtain that $d_1=b_1c_1$, as desired.
    \item[Claim 4: $\prod_{i=1^n} b_{p,i} =a$] Observe that $[n]$ is the disjoint union of
    the singleton sets $\{i\}$ for each $i\in [n]$. Using Claim~3 iteratively, we obtain that $\prod_{i=1^n} b_{p,i}$ equals the first element of the pair $p([n])=(a,e)$, proving the claim.     
    \end{description}
    Now we move on to proving (I), (II), and (III). Claims 1 and 4 show that for
    any $n$-ary polymorphism $p\in \pol(\A, \B)$, the tuple $(b_{p,i})_{i=0}^n$
    is an element of $\Mscr_{M,a}$, as stated in (I). As for fact (II), consider
    some $n$-ary polymorphism $p\in \pol(\A, \B)$, a map $\pi: [n]\rightarrow
    [m]$, and the minor $q=p^{(\pi)}$. We will now show that the fact that $\xi$
    preserves minors is equivalent to $b_{q,i}= \prod_{j\in \pi^{-1}(i)}
    b_{p,j}$ for all $i\in [m]$. Indeed, by definition, $b_{q,i}$ is the first element of the pair $q(\{i\}) = p(\pi^{-1}(i))$. However, 
    expressing $\pi^{-1}(i)$ as a disjoint union of singletons and using Claim 3 we obtain that the first element of $p(\pi^{-1}(i))$ is the product of the first elements of the pairs
    $p(\{j\})$ for each $j\in \pi^{-1}(i)$, as we wanted to show. \par
    So far we have established that $\xi$ is indeed a minion homomorphism, following (I) and (II). Lastly, we prove that $\xi$ is a bijection, as stated in (III). To see that $\xi$ is surjective, consider a tuple $(b_1,\dots, b_n)\in \Mscr_{M,a}(n)$. Then the map $p: 2^{n} \mapsto \Mscr_{M,a}(n)$ given by 
    \[X \mapsto ( \prod_{i\in X} b_i,  \prod_{i\in [n]\setminus X} b_i)
    \]
    is a polymorphism of $(\A, \B)$. Not only that but $\xi(p)= (b_1,\dots, b_n)$. In the other direction, to prove that $\xi$ is injective, we show that an $n$-ary polymorphism 
    $p\in \pol(\A, \B)$ is completely determined by $(b_{p,1}, \dots, b_{p,n})$. Consider an arbitrary set $X\subseteq [n]$, and let $p(X)=(c,d)$. Expressing $X$ as a disjoint union of singletons and using Claim 3 we obtain that $c= \prod_{i\in X} b_{p,i}$. Additionally, by Claim 2 we know that $p([n]\setminus X)=(d,c)$ and using the same argument as before we obtain that $d= \prod_{i\in [n]\setminus X} b_{p,i}$. Thus, $p(X)$
    is completely determined by the tuple $(b_{p,1}, \dots, b_{p,n})$, showing that $\xi$ is injective.     

\section*{Acknowledgements}

We thank the reviewers of the extended abstract~\cite{Larrauri24:icalp} and in
particular of this full version for comments, suggestions, and spotting several typos and mistakes.

\appendix

\section{Reduction to Special Equations}
\label{app:red3}

By definition, instances of $\Lin(S,T)$ can be seen as systems of equations over $S$ where all constants belong to $T$. Any system of equations can be transformed into an equivalent system by adding new variables and breaking larger equations into
smaller ones until every equation is of the form $x_1 x_2 = x_3$, for three variables $x_1,x_2,x_3$, or of the form $x = t$ for a variable $x$ and some constant $t\in T$. 
For instance, the equation
\[
x_1 c_1 x_2 c_2\dots x_k c_k = c_{k+1},
\]
where all $c_1, \dots, c_{k+1}\in T$
can be transformed into the system
\begin{align*}
& x_1 y =  x, \quad \quad y = c_1, \\
& x c_2 \dots x_k c_k =  c_{k+1},
\end{align*}
where $x,y$ are fresh variables. Applying these steps in succession yields a system where every equation is of the desired form. \par
When considering equations over groups the same idea works, but one needs to take into account inverted variables $x^{-1}$. Given a system of equations over a group $G$ with constants in a subgroup $H\substr G$, we can substitute any instance of the inverted variable $x^{-1}$ 
with instances of a fresh variable $y$ after adding the equations $xy=z$, $z=e$
to the system, where $z$ is another fresh variable. \par

\section{Dichotomies For Equations over Monoids and Groups}
\label{app:csps}

In this section we classify the complexity of $\Lin(G)$ for a group $G$
and $\Lin(M)$ for a monoid $M$ as corollaries of the Dichotomy Theorem for CSPs. These results where obtained previously in~\cite{Goldmann02:ic} and~\cite{KTT07:tcs}. We begin by stating the Dichotomy Theorem. 
A $n$-ary map $p:A^n\to A$ is called \emph{cyclic} if 
$p(a_1, a_2,\dots, a_n)= p(a_n, a_1, \dots, a_{n-1})$ for all $a_1,\dots, a_n\in
A$. \par

\begin{theorem}[\cite{Bulatov05:classifying,Bulatov17:focs,Zhuk20:jacm}]
\label{thm:dichotomy}
Let $\A$ be a finite relational structure. Then $\csp(\A)$ is tractable if
  $\pol(\A)$ contains a cyclic polymorphism of arity at least $2$. Otherwise
  $\csp(\A)$ is \NP-hard.     
\end{theorem}

\begin{theorem}\label{thm:dichotomy-eqns}
    Let $G$ be a group. Then $\Lin(G)$ has a cyclic polymorphism $p$ of arity at least $2$ if and only if $G$ is Abelian. 
\end{theorem}
\begin{proof}
    Suppose that $G$ is Abelian. Let $n\geq 2$ be such that $g^n=g$ 
    for all $g\in G$. Then the $n$-ary map $p:G^n \rightarrow G$ 
    given by $(g_1,\dots, g_n) \mapsto \prod_{i=1}^n g_i$ is a cyclic polymorphism of $\Lin(G)$. \par
    In the other direction, suppose that $p:G^n \rightarrow G$ 
    is a cyclic polymorphism of $\Lin(G)$ and $n\geq 2$. 
    Let $g_1,g_2\in G$. As $p$ is idempotent, we have that
    $g_1 = p(g_1,g_1,\dots, g_1)$.
    Using the fact that $p$ is a group homomorphism and is cyclic, we obtain that
    \[
    p(g_1,g_1,\dots, g_1)= 
    p(g_1,e,\dots, e) p(e, g_1, \dots, e) \cdots
    p(e,e, \dots, g_1) = p(g_1^n, e, \dots, e).
    \]
    Similarly, we can obtain that $g_2=p(e,g_2^n,\dots, e)$.
    This way, 
    \[
    g_1 g_2  = p(g_1^n, g_2^n, \dots, e)=
    g_2 g_1.
    \]
    As our initial choice of $g_1,g_2$ was arbitrary, this proves that $G$ is Abelian. 
\end{proof}

\begin{theorem}
    Let $M$ be a monoid. Then $\Lin(M)$ has a cyclic polymorphism $p$ of arity at least $2$ if and only if $M$ is Abelian and regular. 
\end{theorem}
\begin{proof}
  In one direction, if $M$ is regular then,
  by the second item in~\Cref{lem:regularity}
  there is some $n\geq 2$ satisfying $g^n=g$ for all $g\in M$. This way, we can define a cyclic polymorphism of $\Lin(M)$ 
  exactly as in the proof of~\Cref{thm:dichotomy-eqns}. \par
  In the other direction, let $p:M^n \rightarrow M$ be a $n$-ary cyclic
  polymorphism of $\Lin(M)$, where $n\geq 2$. By the same arguments as in the
  the proof of~\Cref{thm:dichotomy-eqns}, $M$ must be Abelian. Let us show now that $M$ is regular. Let $\psi: M \rightarrow M$ be the homomorphism
  given by $g \mapsto p(g,e,\dots, e)$. Using the fact that $p$ is cyclic and idempotent it must hold that
  $g = \psi(g)^n$. This shows that $\psi$ is a bijection and that $\psi(g)
  \geqsg g$ for all $g\in M$. As $M$ is finite, the only way this is possible is
  that $\psi(g) \sim g$ for all $g$. However, from $\psi(g)^n=g$ and $n\geq 2$
  we deduce that $\psi(g)\sim \psi(g)^2$, so $\psi(g)$ is a regular element by
  the fourth item in~\Cref{lem:regularity}. This holds for an arbitrary $g$, so every element in $M$ is regular. 
\end{proof}

{\small
\bibliographystyle{plainurl}
\bibliography{lz}
}

\end{document}